\RequirePackage{fix-cm}
\documentclass[smallextended]{svjour3}       
\smartqed  
\usepackage{graphicx}
\usepackage{}
\usepackage{amssymb}[1995/11/28]

\usepackage{amsfonts}
\usepackage{amsmath,amssymb}
\usepackage{bbm}
\usepackage{mathrsfs}
\usepackage{graphicx}
\usepackage{ntheorem}

\newcommand{\bm}[1]{\mbox{\boldmath{$#1$}}}

\begin{document}

\title{Yang-Yang functions, Monodromy and knot polynomials
}


\author{Peng Liu         \and
        Wei-Dong Ruan 
}


\institute{Peng Liu \at
              School Of Applied Science, Beijing Information Science and Technology University, Beijing,10010,China\\
              \email{pliu@mail.ustc.edu.cn}           
           \and
           Wei-Dong Ruan \at
              Institute of Mathematics, Academy of Mathematics and Systems Science, Chinese Academy of Sciences, Beijing, 100190, China
}

\date{Received: date / Accepted: date}

\maketitle

\begin{abstract}
We derive a structure of $\mathbb{Z}[t,t^{-1}]$-module bundle from a family of Yang-Yang functions. For the fundamental representation of the complex simple Lie algebra of classical type, we give explicit wall-crossing formula and prove that the monodromy representation of the $\mathbb{Z}[t,t^{-1}]$-module bundle is equivalent to the braid group representation induced by the universal R-matrices of $U_{h}(g)$. We show that two transformations induced on the fiber by the symmetry breaking deformation and respectively the rotation of two complex parameters commute with each other.
\end{abstract}

\section{Introduction}
\label{sec:intro}
Yang-Yang function was named by N. Nekrasov, A. Rosly and S. Shatashvili in \cite{Nekrasov}. Originated from C. N. Yang and C. P. Yang's paper \cite{Yang66one}\cite{Yang69}, it was used for the analysis of the non-linear Schr\"odinger model. Behind this function hides a quantum integrable system \cite{Nekrasov}\cite{V11}, thus it aroused much interest. Yang-Yang function can also be realized as the power in the correlation function of the free field realization of Virasoro vertex operators \cite{GW}. D. Gaitto and E. Witten used this realization to derive Jones polynomial for knots from four dimensional Chern-Simons gauge theory. Our interest is to figure out the structure underlying the derivation of knot invariants from the Yang-Yang function. Similar to the work of V. G. Drinfeld and T. Kohno \cite{Chari} \cite{Drinfeld1989c} \cite{Kohno87}, where the monodromy of Knizhnik-Zamolodchikov system was proved to be equivalent to the braid group representation, induced by the universal R-matrices of the quantum enveloping algebra $U_{h}(g)$ of semisimple Lie algebra $g$, we prove that the monodromy representation of the $ \mathbb{Z}[t,t^{-1}]$-module bundle constructed from a family of Yang-Yang functions associated with the fundamental representation of the classical complex simple Lie algebra is equivalent to the braid group representation, induced by the universal R-matrices of $U_{h}(g)$. Furthermore, transformations induced on the fiber by two parameter deformations, the symmetry breaking parameter $c$ from $c=0$ to $c\rightarrow \infty$ and respectively the rotation of two singular complex parameters $z_{1}$ and $z_{2}$, commute with each other. By studying the monodromy, creation and annihilation matrices from the parameter deformations, one can derive the HOMFLY-PT polynomial and Kauffman polynomial for knots. We expect the existence of new knot invariants, different from HOMFLY-PT and Kauffman polynomials, from general representations of Lie algebras. It will be investigated elsewhere.

This paper is organized as follows. Firstly we give the definition of Yang-Yang function in general case and derive from it the $ \mathbb{Z}[t,t^{-1}]$-module bundle structure, then state two main theorems in section 2. In section 3, the parameter rotation for Yang-Yang function is considered. Four types of variations of critical point of Yang-Yang function are studied in detail. For the fundamental representation of classical complex simple Lie algebra, general wall-crossing formula and monodromy representation are derived. The first main theorem follows. In section 4, we discuss the relation between two monodromy representations of $c=0$ and $c\rightarrow +\infty$, which leads to a proof of the second main theorem.

\section{Yang-Yang function, monodromy of its associated bundle and main theorems} \label{sec:main}

\subsection{Yang-Yang function and its critical point}
Let $g$ be a finite dimensional complex simple Lie algebra with Chevalley generators $\{E_{i},F_{i},H_{i}\},i=1,...,n$ and Cartan matrix $C_{ij}=\frac{2(\alpha_{i},\alpha_{j})}{(\alpha_{j},\alpha_{j})}$, where $\alpha_{i}$ primary roots of $g$ and $(,)$ the inner product induced from Killing form. Let $P^{+}$ be the positive Weyl chamber in the weight space of $g$. Fundamental weights $\omega_{i}$ are a set of bases of weight space satisfying $\frac{2(\omega_{i},\alpha_{j})}{(\alpha_{j},\alpha_{j})}=\delta_{ij}$. Weyl vector is defined as $\rho=\sum^{n}_{i=1}\omega_{i}$, thus $\frac{2(\rho,\alpha_{i})}{(\alpha_{i},\alpha_{i})}=1$ for any $i$ $(1\leq i\leq n)$. Let ${\bm \lambda}=(\lambda_{1},...,\lambda_{m})$ be a sequence of weights $\lambda_{a}\in P^{+}$. Denote by $V_{{\bm\lambda}}$ the representation $V_{\lambda_{1}}\otimes...\otimes V_{\lambda_{m}}$ and $\Omega_{{\bm \lambda}}$ the set of weights in $V_{{\bm \lambda}}$. ${\bm l}=(l_{1},...,l_{n})$ a sequence of nonnegative integers is just a $n$ partition of $l=\sum_{i=1}^{n}l_{i}$ and ${\bm 0}=(0,...,0)$.

Yang-Yang function is defined by
\begin{equation}
\begin{split}
{\bm W}({\bm w},{\bm z},{\bm\lambda},{\bm l})=&\sum_{j=1}^{l}\sum _{a=1}^{m}(\alpha_{i_{j}},\lambda _{a})\log(w_{j}-z_{a}) -\sum _{1\leq j<k\leq l}(\alpha_{i_{j}},\alpha_{i_{k}})\log(w_{j}-w_{k})\\
&-\sum _{1\leq a< b\leq m}(\lambda_{a},\lambda_{b})\log(z_{a}-z_{b}).
\end{split}
\end{equation}

It is a function of complex variables ${\bm w}=(w_{1},...,w_{l})$, distinct complex parameters ${\bm z}=(z_{1},...,z_{m})$, weights ${\bm\lambda}$ and $n$ partition ${\bm l}$ of $l$.
To each $z_{a}$ and $w_{j}$, a dominant integral weight $\lambda_{a}\in P^{+}$ and a primary root
$\alpha_{i_{j}}$ are associated respectively, where $i_{j}\in \{ 1,...,n\}$. Let $l_{k}=\#\{j\mid 1\leq j\leq l, i_{j}=k\}$ be the number of $\alpha_{k}$ in $\{\alpha_{i_{j}}\}_{j=1,...,l}$.
Denote by $\alpha({\bm l})$ the sum of all these primary roots $\alpha({\bm l})=l_{1}\alpha_{1}+...l_{n}\alpha_{n}$. Then ${\bm W}({\bm w},{\bm z},{\bm \lambda},{\bm 0})$ is a constant function of ${\bm w}$.

The critical points of ${\bm W}({\bm w},{\bm z},{\bm\lambda},{\bm l})$ satisfy
\begin{equation}
  \frac{\partial {\bm W}({\bm w},{\bm z},{\bm \lambda},{\bm l})}{\partial w_{j}}=0, j=1,...,l
\end{equation} equivalently,
\begin{equation}
  \sum_{a}\frac{(\alpha_{i_{j}},\lambda_{a})}{w_{j}-z_{a}}=\sum_{s\neq j}\frac{(\alpha_{i_{j}},\alpha_{i_{s}})}{w_{j}-w_{s}}, j=1,...,l.
\end{equation}
By definition, if ${\bm w}$ is a solution of $\frac{\partial {\bm W}}{\partial w_{j}}=0$ and $(\alpha_{i_{j}},\alpha_{i_{s}})\neq 0$, then $w_{j}\neq w_{s}$. Also if $(\alpha_{i_{j}},\lambda_{a})\neq 0$, then $w_{j}\neq z_{a}$. Under the permutation of all the coordinates $w_{j}$ associated with the same primary root, the critical point equation above is invariant, thus we do not distinguish these critical points.

The critical points of Yang-Yang function have a close relation with the singular vectors of the tensor product space $V_{{\bm \lambda}}=V_{\lambda_{1}}\otimes V_{\lambda_{2}}\otimes ...\otimes V_{\lambda_{m}}$. Let $Sing V=\{v\in V \mid E_{i} v=0, \forall i \}$ be the subspace of singular vectors in $V$. The following fact is known.

\begin{theorem}\label{LMV} If $\lambda_{1}+...+\lambda_{m}-\alpha({\bm l})\in P^{+}$, then for all nondegenerate critical points of Yang-Yang function ${\bm W}({\bm w},{\bm z},{\bm \lambda},{\bm l})$, there exists a singular vector in $Sing V_{\bm \lambda}$.
\end{theorem}
If $\lambda_{1}+...+\lambda_{m}-\alpha({\bm l})\in P^{+}$, for each critical point of Yang-Yang function, Bethe vector can be constructed\cite{MV05norm}.   Shapovarov form of Bethe vector is proved to be equal to Hessian of Yang-Yang function at the critical point\cite{V11}. Thus nondegeneration of critical point guarantees Bethe vector is nonzero. By theorem 11.1 in \cite{RV94}, each nonzero Bethe vector belongs to $Sing V_{{\bm \lambda}}$. For details of proof, we refer to \cite{V11}\cite{MV05norm}\cite{RV94}. The master function in the reference is just the exponential of Yang-Yang function ${\bm W}({\bm w},{\bm z},{\bm \lambda},{\bm l})$.

Formally denote by $w_{0}$ the critical point of constant function${\bm W}({\bm w},{\bm z},{\bm \lambda},{\bm 0})$.

\begin{theorem}\label{1sing} Let $g$ be a classical complex simple Lie algebra. If ${\bm z}=(z_{1},z_{2})$, ${\bm \lambda}=(\omega_{1},\omega_{1})$, the nondegenerate critical points of the family of Yang-Yang functions $$\{{\bm W}({\bm w},{\bm z},{\bm \lambda},{\bm l})\mid V_{2\omega_{1}-\alpha({\bm l})}  \subset  V_{{\bm \lambda}}\}$$ together with  the formal degenerate critical point $w_{0}$ of ${\bm W}({\bm w},{\bm z},{\bm \lambda},{\bm 0})$ one to one correspond to the weight vectors in $SingV_{{\bm \lambda}}.$
\end{theorem}
\begin{proof} By Littelmann-Littlewood-Richardson rule\cite{littelmann97}, the following direct sum decompositions can be derived for each case of classical Lie algebras.
\begin{equation}
\begin{split}
A_{n}:&\quad V_{\omega_{1}}\otimes V_{\omega_{1}}=V_{2\omega_{1}}\oplus V_{2\omega_{1}-\alpha_{1}}\\
B_{n}:&\quad V_{\omega_{1}}\otimes V_{\omega_{1}}=V_{2\omega_{1}}\oplus V_{2\omega_{1}-\alpha_{1}}\oplus V_{2\omega_{1}-2\alpha_{1}-...-2\alpha_{n}}\\
C_{n}:&\quad V_{\omega_{1}}\otimes V_{\omega_{1}}=V_{2\omega_{1}}\oplus V_{2\omega_{1}-\alpha_{1}}\oplus  V_{2\omega_{1}-2\alpha_{1}-...-2\alpha_{n-1}-\alpha_{n}}\\
D_{n}:&\quad V_{\omega_{1}}\otimes V_{\omega_{1}}=V_{2\omega_{1}}\oplus V_{2\omega_{1}-\alpha_{1}}\oplus  V_{2\omega_{1}-2\alpha_{1}-...-2\alpha_{n-2}-\alpha_{n-1}-\alpha_{n}}\\
\end{split}
\end{equation}
For $B_{n}$ Lie algebra, the singular vectors of $Sing V_{\omega_{1}}\otimes V_{\omega_{1}}$ are $$v_{2\omega_{1}},v_{2\omega_{1}-\alpha_{1}},v_{2\omega_{1}-2\alpha_{1}-...-2\alpha_{n}}.$$
By theorem \ref{LMV}, there is an injective map between the set of the nondegenerate critical points and $SingV_{\omega_{1}}\otimes V_{\omega_{1}}$.
Corresponding to each singular vector, the explicit solution of the critical point equation can be derived by lemma in section \ref{subsct:monodromy}. They are nondegenerate except when $l=0$. Thus the map is bijective on $SingV_{\omega_{1}}\otimes V_{\omega_{1}}\backslash v_{2\omega_{1}}$. When $l=0$, the formal degenerate critical point $w_{0}$ corresponds to the singular vector $v_{2\omega_{1}}$ in $SingV_{\omega_{1}}\otimes V_{\omega_{1}}$. Therefore, one to one correspondence is proved. For $A_{n}$, $C_{n}$, $D_{n}$, the proofs are similar.
\end{proof}

With an additional deformation parameter $c\in \mathbb{R}$, symmetry breaking Yang-Yang function is defined by

\begin{equation}
\begin{split}
{\bm W}_{c}({\bm w},{\bm z},{\bm \lambda},{\bm l})=&\sum _{j,a}(\alpha_{i_{j}},\lambda _{a})\log(w_{j}-z_{a}) -\sum _{j<k}(\alpha_{i_{j}},\alpha_{i_{k}})\log(w_{j}-w_{k})\\
&-\sum _{a< b}(\lambda_{a},\lambda_{b})\log(z_{a}-z_{b})-c\sum _{j}(\rho,\alpha_{i_{j}})w_{j}+c\sum _{a}(\rho,\lambda_{a})z_{a}.
\end{split}
\end{equation}

It is clear that ${\bm W}_{0}({\bm w},{\bm z},{\bm \lambda},{\bm l})={\bm W}({\bm w},{\bm z},{\bm \lambda},{\bm l})$. In fact, when $c\in \mathbb{Z}_{\geq0}$, ${\bm W}_{c}({\bm w},{\bm z},{\bm \lambda},{\bm l})$ is just a limitation of the parameter deformation of $\eta$ with an additional parameter $z_{\infty}=1$ and the dominant integral weight $\lambda_{\infty}=c\rho\in P^{+}$ associated with it.
\begin{equation}\label{lpt}
\begin{split}
&{\bm W}_{c}({\bm w},{\bm z},{\bm \lambda},{\bm l})\\
=&\lim_{\eta\rightarrow0} [\sum _{j,a}(\alpha_{i_{j}},\lambda _{a})\log(w_{j}-z_{a}) +\sum _{j}(\frac{c\rho}{\eta},\alpha_{i_{j}})\log(1-\eta w_{j})\\
&-\sum _{j<k}(\alpha_{i_{j}},\alpha_{i_{k}})\log(w_{j}-w_{k})-\sum _{0<a< b}(\lambda_{a},\lambda_{b})\log(z_{a}-z_{b}))\\
&-\sum _{a}(\frac{c\rho}{\eta},\lambda_{a})\log(1-\eta z_{a}) ].
\end{split}
\end{equation}

The critical point equation of ${\bm W}_{c}({\bm w},{\bm z},{\bm \lambda},{\bm l})$ is
\begin{equation}\label{CE2}
\sum^{m}_{a=1}\frac{(\alpha_{i_{j}},\lambda_{a})}{w_{j}-z_{a}}=\sum_{s\neq j}\frac{(\alpha_{i_{j}},\alpha_{i_{s}})}{w_{j}-w_{s}}+c(\rho,\alpha_{i_{j}}), j=1,...,l.
\end{equation}

\begin{lemma}\label{split}
If $w_{j}$ is the coordinate of the critical point of ${\bm W}_{c}({\bm w},{\bm z},{\bm \lambda},{\bm l})$, then $$\lim_{c\rightarrow+\infty} w_{j}\in \{z_{a}\}_{a=1,...,m}.$$
\end{lemma}
\begin{proof}
Because $(\rho,\alpha_{i_{j}})>0$, it is clear for any $j$, $$\lim_{c\rightarrow +\infty} w_{j}\in \{z_{a}\}_{a=1,...,m}\cup \{w_{k}\}_{k\neq j}.$$
Divide the set $\{w_{1},...,w_{l}\}$ into a disjoint union of $Z_{1},...,Z_{m}$ and $M$, where $Z_{a}, a=1,...,m$ contains the coordinates $w_{j}$ which tend to $z_{a}$ and $M$ the rest of them. Assume $M=\{w_{m_{1}},...,w_{m_{p}}\}\neq \emptyset$, then $p>1$. Sum the equations of $j\in M$
\begin{equation}\label{sum}
\sum_{j=m_{1}}^{m_{p}} \sum _{a}\dfrac {(\alpha_{i_{j}},\lambda_{a})} {w_{j}-z_{a}}=\sum_{j=m_{1}}^{m_{p}} \sum _{s\neq j}\dfrac {(\alpha_{i_{j}},\alpha_{i_{s}})} {w_{j}-w_{s}}+c\sum_{j=m_{1}}^{m_{p}}(\rho,\alpha_{i_{j}}).
\end{equation}
When $c\rightarrow +\infty$, the left hand side of \eqref{sum} is bounded.
$$\sum_{j=m_{1}}^{m_{p}} \sum _{s\neq j}\dfrac {(\alpha_{i_{j}},\alpha_{i_{s}})} {w_{j}-w_{s}}=\sum_{w_{j}\in M} \sum _{w_{s}\in Z_{1}\cup...\cup Z_{m}}\dfrac {(\alpha_{i_{j}},\alpha_{i_{s}})} {w_{j}-w_{s}},$$ where the summation in $M$ is canceled, thus also bounded.
\eqref{sum} leads to contradiction. Therefore, $M=\emptyset$. The lemma follows.
\end{proof}

Denote also by $w_{0}$ the critical point of the constant function${\bm W}_{c}({\bm w},{\bm z},{\bm \lambda},{\bm 0})$. For ${\bm z}=z_{1}$ and ${\bm \lambda}=\omega_{1}$, we have the following theorem

\begin{theorem}\label{1prm} Let $g$ be a classical complex simple Lie algebra. If $c\in \mathbb{Z}_{\geq0}$ and $c\geq2$, the nondegenerate critical points of the family of Yang-Yang functions $$\{{\bm W}_{c}({\bm w},z_{1},\omega_{1},{\bm l})\}_{ \omega_{1}-\alpha({\bm l})\in \Omega_{{\omega_{1}}}}$$ together with  the formal degenerate critical point $w_{0}$ one to one correspond to the weight vectors in $V_{\omega_{1}}.$
\end{theorem}
\begin{proof}

By Littelmann-Littlewood-Richardson rule\cite{littelmann97}, if $c\in \mathbb{Z}_{\geq0}$ and $c\geq2$, the following direct sum decompositions can be derived.
\begin{equation}
\begin{split}
A_{n}: V_{\omega_{1}}\otimes V_{c\rho}=&V_{\omega_{1}+c\rho}\oplus V_{\omega_{1}+c\rho-\alpha_{1}}\oplus ...\oplus  V_{\omega_{1}+c\rho-\alpha_{1}-...-\alpha_{n}}\\
B_{n}: V_{\omega_{1}}\otimes V_{c\rho}=&V_{\omega_{1}+c\rho}\oplus V_{\omega_{1}+c\rho-\alpha_{1}}\oplus ...\oplus  V_{\omega_{1}+c\rho-\alpha_{1}-...-\alpha_{n}}\oplus \\
&V_{\omega_{1}+c\rho-\alpha_{1}-...-2\alpha_{n}}\oplus V_{\omega_{1}+c\rho-\alpha_{1}-...-2\alpha_{n-1}-2\alpha_{n}}
 \oplus \\&...\oplus  V_{\omega_{1}+c\rho-2\alpha_{1}-...-2\alpha_{n-1}-2\alpha_{n}}\\
C_{n}: V_{\omega_{1}}\otimes V_{c\rho}=&V_{\omega_{1}+c\rho}\oplus V_{\omega_{1}+c\rho-\alpha_{1}}\oplus ...\oplus  V_{\omega_{1}+c\rho-\alpha_{1}-...-\alpha_{n}}\oplus  \\ &V_{\omega_{1}+c\rho-\alpha_{1}-...-2\alpha_{n-1}-\alpha_{n}}
\oplus ...\oplus  V_{\omega_{1}+c\rho-2\alpha_{1}-...-2\alpha_{n-1}-\alpha_{n}}\\
D_{n}: V_{\omega_{1}}\otimes V_{c\rho}=&V_{\omega_{1}+c\rho}\oplus V_{\omega_{1}+c\rho-\alpha_{1}}\oplus V_{\omega_{1}+c\rho-\alpha_{1}-\alpha_{2}}\oplus ...\oplus \\ &V_{\omega_{1}+c\rho-\alpha_{1}-...-\alpha_{n}} \oplus V_{\omega_{1}+c\rho-\alpha_{1}-...-2\alpha_{n-2}-\alpha_{n-1}-\alpha_{n}}
\oplus\\& ...\oplus  V_{\omega_{1}+c\rho-2\alpha_{1}-...-2\alpha_{n-2}-\alpha_{n-1}-\alpha_{n}}\\
\end{split}
\end{equation}

For $A_{n}$, it is clear that the singular vectors of $SingV_{\omega_{1}}\otimes V_{c\rho}$ $$v_{\omega_{1}+c\rho},v_{\omega_{1}+c\rho-\alpha_{1}},...,v_{\omega_{1}+c\rho-\alpha_{1}-...-\alpha_{n}},$$
one to one correspond to the weight vectors $$v_{\omega_{1}},v_{\omega_{1}-\alpha_{1}},...,v_{\omega_{1}-\alpha_{1}-...-\alpha_{n}},$$  in $V_{\omega_{1}}$.
By theorem \ref{LMV}, there is an injective map between the set of the nondegenerate critical points and $SingV_{\omega_{1}}\otimes V_{c\rho}$.
Corresponding to each singular vectors, ${\bm l}$ satisfies the admissible condition and the explicit solutions of critical points with two parameters $z_{1}=0, z_{2}=1$ are given in \cite{LMV16}. They are nondegenerate except when $l=0$. Thus the map is bijective on $SingV_{\omega_{1}}\otimes V_{c\rho}\backslash v_{\omega_{1}+c\rho}$. When $l=0$, the formal degenerate critical point $w_{0}$ corresponds to the singular vector $v_{\omega_{1}+c\rho}$ in $SingV_{\omega_{1}}\otimes V_{c\rho}$. By \eqref{lpt}, the critical point of ${\bm W}_{c}({\bm w},z_{1},\omega_{1},{\bm l})$ are the limitation of the parameter deformation of those solutions. We give the explicit expressions of them in the lemma \ref{lem:exsol} in section \ref{subsubsct:B}. Therefore, one to one correspondence is proved. For $B_{n}$, $C_{n}$, $D_{n}$, the proofs are similar.
\end{proof}

\begin{remark} Except for $B_{n}$, the direct sum decompositions above are also correct if $c=1$ .\end{remark}

\subsection{$V_{\lambda}\otimes V_{\lambda}$ realized as the fiber of $\mathbb{Z}[t,t^{-1}]$-module bundle}

For any dominant integral weight $\lambda \in P^{+}$, there exists a linear automorphism of $V_{\lambda}\otimes V_{\lambda}$ called R-matrix:  $$R:V_{\lambda}\otimes V_{\lambda}\rightarrow V_{\lambda}\otimes V_{\lambda}$$ satisfying Yang-Baxter equation \cite{Kassel}
$$(R\otimes id_{V_{\lambda}})(id_{V_{\lambda}}\otimes R)(R\otimes id_{V_{\lambda}})=(id_{V_{\lambda}}\otimes R)(R\otimes id_{V_{\lambda}})(id_{V_{\lambda}}\otimes R).$$

To study R-matrix, $m=2$ is sufficient. In the following sections, we consider the fundamental representation $\lambda=\omega_{1}$ of complex simple Lie algebra of classical types. Assume that $z_{1}$ and $z_{2}$ have the same real part and $Im z_{1}> Im z_{2}$.

\subsubsection{$V_{\lambda}\otimes V_{\lambda}$ and the thimble space of Yang-Yang functions}

By lemma \ref{split}, when $c\rightarrow +\infty$, equations of \eqref{CE2} splits into two separated sets involving only $z_{1}$ or $z_{2}$ and its nondegenerate solutions space is just the tensor product of two nodegenerate solutions spaces with only one parameter $z_{1}$ or $z_{2}$. Therefore, by theorem\ref{1prm}, when $c\rightarrow +\infty$, the set of nondegenerate critical points of \eqref{CE2} together with formal critical point $w_{0}$ one to one correspond to the weight vectors in $V_{\omega_{1}}\otimes V_{\omega_{1}}$.

The thimble of ${\bm W}_{c}({\bm w},{\bm z},{\bm \lambda},{\bm l})$, defined \cite{HL1} as the cycle formulated by the gradient flow of the real part of Yang-Yang function starting from the critical point is a $l$ dimensional manifold. Therefore, for all vectors $v_{\omega_{1}-\alpha({\bm k})}\otimes v_{\omega_{1}-\alpha({\bm l})}\in V_{\omega_{1}}\otimes V_{\omega_{1}}$, there exists a unique thimble of ${\bm W}_{c}({\bm w},{\bm z},{\bm \lambda},{\bm k+\bm l})$ denoted by $J_{{\bm k},{\bm l}}$ and $J_{{\bm k},{\bm l}}=J_{{\bm k}}\times J_{{\bm l}}$, where $J_{{\bm k}}$ and $J_{{\bm l}}$ are thimbles of ${\bm W}_{c}({\bm w},z_{1},\omega_{1},{\bm k})$ and ${\bm W}_{c}({\bm w},z_{2},\omega_{1},{\bm l})$ respectively. Therefore, we have

\begin{theorem}\label{thmsb2c}When $c\rightarrow +\infty$, the basis of $V_{\omega_{1}}\otimes V_{\omega_{1}}$ one to one correspond to thimbles generated from a family of Yang-Yang functions $$\{{\bm W}_{c}({\bm w},{\bm z},{\bm \lambda},{\bm k+\bm l})\}_{ \omega_{1}-\alpha({\bm k}),\omega_{1}-\alpha({\bm l})\in \Omega_{{\omega_{1}}}}.$$
\end{theorem}
Note that the solution of critical point equation \eqref{CE2} is not unique. It means that there exist different thimbles corresponding to different critical solutions of the same Yang-Yang function. We will see it in the example of section \ref{subsec:ex}.

\subsubsection{The thimble space as a fiber bundle of $\mathbb{Z}[t,t^{-1}]$-module}

As in \cite{ATY}, \cite{FFR} and \cite{Frenkel95}, Yang-Yang function appears naturally as an exponent in the correlation function of Wakimoto realization of Kac-Moody algebra at arbitrary level $\kappa$ :
\begin{equation}
\begin{split}
&\int _{\Gamma }\prod_{j,a}(w_{j}-z_{a})^{-\frac{(\alpha_{i_{j}},\lambda_{a})}{\kappa+h^{\vee}}}\prod_{j<s}(w_{j}-w_{s})^{\frac{(\alpha_{i_{j}},\alpha_{i_{s}})}{\kappa+h^{\vee}}}(z_{1}-z_{2})^{\frac{(\lambda,\lambda)}{\kappa+h^{\vee}}}\prod _{j}dw_{j}\\
=&\int _{\Gamma }e^{-\frac{{\bm W}({\bm w},{\bm z},{\bm \lambda},{\bm l})}{\kappa+h^{\vee}}}\prod _{j}dw_{j},
\end{split}
\end{equation}
where $h^{\vee}$ is the dual Coxeter number. The crucial problem is to figure out the transformation induced on the thimble space by parameter deformation of $e^{-\frac{{\bm W}_{c}({\bm w},{\bm z},{\bm \lambda},{\bm l})}{\kappa+h^{\vee}}}$, which leads us to study the extra structure of the thimble space.

Thimbles of the real part of Yang-Yang function ${\bm W}_{c}({\bm w},{\bm z},{\bm \lambda},{\bm l})$ are the same as thimbles of $\mid e^{-\frac{{\bm W}_{c}({\bm w},{\bm z},{\bm \lambda},{\bm l})}{\kappa+h^{\vee}}}\mid$, except that there are infinite number of pre-images with difference $2\pi \mathbbm{i}$ in their imaginary parts. Along the thimble of its real part, conservation law \cite{HL1}
of the imaginary part of Yang-Yang function implies global invariance of the phase factor $e^{- \frac{Im {\bm W}_{c}({\bm w},{\bm z},{\bm \lambda},{\bm l})}{\kappa+h^{\vee}}\mathbbm{i}}$. By theorem \ref{thmsb2c}, when $c\rightarrow +\infty$, the basis of $V_{\omega_{1}}\otimes V_{\omega_{1}}$  one to one correspond to the family of thimbles $\{J_{{\bm k},{\bm l}}\}_{ \omega_{1}-\alpha({\bm k}),\omega_{1}-\alpha({\bm l})\in \Omega_{{\omega_{1}}}}$ of $\{\mid e^{-\frac{{\bm W}_{c}({\bm w},{\bm z},{\bm \lambda},{\bm k+\bm l})}{\kappa+h^{\vee}}}\mid\}_{ \omega_{1}-\alpha({\bm k}),\omega_{1}-\alpha({\bm l})\in \Omega_{{\omega_{1}}}}$. Each of them has infinite pre-images with different global phases.
Choose a branch from infinite pre-images for each thimble. Define $q=e^{\frac{2\pi \mathbbm{i}}{\kappa+h^{\vee}}}$. Continuous deformation of the parameters in $e^{-\frac{{\bm W}_{c}({\bm w},{\bm z},{\bm \lambda},{\bm l})}{\kappa+h^{\vee}}}$ will induce a global phase factor variation along the thimble of $\mid e^{-\frac{{\bm W}_{c}({\bm w},{\bm z},{\bm \lambda},{\bm l})}{\kappa+h^{\vee}}}\mid$.
$\{J_{{\bm k},{\bm l}}\}_{ \omega_{1}-\alpha({\bm k}),\omega_{1}-\alpha({\bm l})\in \Omega_{{\omega_{1}}}}$ are elements of a certain relative homology group and have a natural integral structure \cite{GW}. Let $ \mathbb{Z}[t,t^{-1}]$ be the ring of Laurrent polynomials of $t$ over $\mathbb{Z}$, then these variations naturally make thimble space a module over ring $ \mathbb{Z}[t,t^{-1}]$, where $t$ is some minimal power of $q$ during the variation. From now on, we consider $V_{\omega_{1}}\otimes V_{\omega_{1}}$ as a $\mathbb{Z}[t,t^{-1}]$-module generated by the basis $\{v_{\omega_{1}-\alpha({\bm k})}\otimes v_{\omega_{1}-\alpha({\bm l})}\}_{ \omega_{1}-\alpha({\bm k}),\omega_{1}-\alpha({\bm l})\in \Omega_{{\omega_{1}}}}$. Let $\mathfrak{J}=\mathbb{Z}[t,t^{-1}]\{J_{{\bm k},{\bm l}}\}_{ \omega_{1}-\alpha({\bm k}),\omega_{1}-\alpha({\bm l})\in \Omega_{{\omega_{1}}}}$ be the module over $\mathbb{Z}[t,t^{-1}]$ generated by the corresponding thimbles, then we have a natural decomposition of $V_{\omega_{1}}\otimes V_{\omega_{1}}$ as a direct sum of  $\mathbb{Z}[t,t^{-1}]$-submodules.

 \begin{lemma}
\begin{equation}\label{11crspd}
  V_{\omega_{1}}\otimes V_{\omega_{1}}\cong \mathfrak{J}=\underset{\mathbbm{q}}{\oplus} E^{\mathbbm{q}},
\end{equation}
 where $E^{\mathbbm{q}}=\mathbb{Z}[t,t^{-1}]\{J_{{\bm k},{\bm l}}\}_{k+l=\mathbbm{q}}$ is a $\mathbb{Z}[t,t^{-1}]$-submodule generated by all the $\mathbbm{q}$ dimensional thimbles of $\mathfrak{J}$.
\end{lemma}
Because $\mathfrak{J}$ depends on the parameters $(z_{1},z_{2})$, there exists a bundle $\mathfrak{J}\overset{\pi}{\rightarrow} X_{2}$, where
\begin{equation}\label{Confi}
  X_{2}=\{(z_{1},z_{2})\in \mathbb{C}^{2}\mid z_{1}\neq z_{2}\}/(z_{1},z_{2})\sim(z_{2},z_{1})
\end{equation}
is the configuration space of the two complex parameters $z_{1}$ and $z_{2}$. The fiber $\pi^{-1}(z_{1},z_{2})=\mathfrak{J}({\bm z})$ is defined to be a $\mathbb{Z}[t,t^{-1}]$-module generated by all the thimbles of $\{\mid e^{-\frac{{\bm W}_{c}({\bm w},{\bm z},{\bm \lambda},{\bm k+\bm l})}{\kappa+h^{\vee}}}\mid\}_{\omega_{1}-\alpha({\bm k}),\omega_{1}-\alpha({\bm l})\in \Omega_{{\omega_{1}}}}$.
Let $E^{\mathbbm{q}}\overset{\pi^{\mathbbm{q}}}{\rightarrow} X_{2}$ be the $\mathbb{Z}[t,t^{-1}]$-module sub-bundle of $\mathfrak{J}\overset{\pi}{\rightarrow} X_{2}$. Because of \eqref{11crspd}, we have the following decomposition of the $\mathbb{Z}[t,t^{-1}]$-module bundle.

\begin{lemma} $$\pi=\underset{\mathbbm{q}}{\oplus}\pi^{\mathbbm{q}}.$$ \end{lemma}

In sum, the thimble space of $\{e^{-\frac{{\bm W}_{c}({\bm w},{\bm z},{\bm \lambda},{\bm k+\bm l})}{\kappa+h^{\vee}}}\}_{\omega_{1}-\alpha({\bm k}),\omega_{1}-\alpha({\bm l})\in \Omega_{{\omega_{1}}}}$ as the fiber of $\mathbb{Z}[t,t^{-1}]$-module  bundle $\mathfrak{J}\overset{\pi}{\rightarrow} X_{2}$ gives $V_{\omega_{1}}\otimes V_{\omega_{1}}$ a geometric realization. We will use it to derive $R$ matrix.

\subsection{Monodromy representation of the $\mathbb{Z}[t,t^{-1}]$-module \\fiber bundle}\label{subsct:bundle}

Fixing a point $P=(z_{1},z_{2})\in X_{2}$, we consider a continuous parameter transformation $T(s): X_{2}\rightarrow X_{2}$ defined by \begin{equation}\label{equ:T}
  T(s)\left(
                                                         \begin{array}{c}
                                                           z_{1} \\
                                                           z_{2} \\
                                                         \end{array}
                                                       \right)=\left(
                                                                 \begin{array}{cc}
                                                                   \frac{1+e^{-\mathbbm{i}\pi s}}{2} & \frac{1-e^{-\mathbbm{i}\pi s}}{2} \\
                                                                   \frac{1-e^{-\mathbbm{i}\pi s}}{2} & \frac{1+e^{-\mathbbm{i}\pi s}}{2} \\
                                                                 \end{array}
                                                               \right)\left(
                                                                        \begin{array}{c}
                                                                          z_{1} \\
                                                                          z_{2} \\
                                                                        \end{array}
                                                                      \right),s\in [0,1].
\end{equation}

It is a clockwise rotation around the middle point of $z_{1}$ and $z_{2}$:
 $$T(1)P=T(1)\left(
                                                         \begin{array}{c}
                                                           z_{1} \\
                                                           z_{2} \\
                                                         \end{array}
                                                       \right)=\left(
                                                         \begin{array}{c}
                                                           z_{2} \\
                                                           z_{1} \\
                                                         \end{array}
                                                       \right)=P.$$
Thus, $T:S^{1}\rightarrow X_{2}$ generates a fundamental group $\pi_{1}(X_{2},P)$ of $X_{2}$ with base point $P$ and also induces a $\mathbb{Z}[t,t^{-1}]$-module transformation $$\sigma\in End(\mathfrak{J}(P),\mathbb{Z}[t,t^{-1}])=End(V_{\omega_{1}}\otimes V_{\omega_{1}},\mathbb{Z}[t,t^{-1}])$$ on the fiber $\mathfrak{J}(P)\cong V_{\omega_{1}}\otimes V_{\omega_{1}}$. $\sigma$ is called monodromy of the bundle and it generates a monodromy group.  Its representation on the fiber space is called monodromy representation. Let $\boldsymbol{B}:V_{\omega_{1}}\otimes V_{\omega_{1}}\rightarrow V_{\omega_{1}}\otimes V_{\omega_{1}}$ be the monodromy representation of the bundle $\mathfrak{J}\overset{\pi}{\rightarrow} X_{2}$ induced by $T:S^{1}\rightarrow X_{2}$ and $\boldsymbol{B_{U_{h}(g)}}$ the braid group representation induced by the universal R-matrices of the quantum enveloping algebra $U_{h}(g)$. Our first main theorem is as following:

\begin{theorem}For the fundamental representation $V_{\omega_{1}}$ of classical complex simple Lie algebra $g$, the monodromy representation $\boldsymbol{B_{YY}}$ of the $\mathbb{Z}[t,t^{-1}]$-module fiber bundle $\mathfrak{J}\overset{\pi}{\rightarrow} X_{2}$ generated by the family of functions $$\{e^{-\frac{{\bm W}_{c}({\bm w},{\bm z},{\bm \lambda},{\bm k+\bm l})}{\kappa+h^{\vee}}}\}_{\omega_{1}-\alpha({\bm k}),\omega_{1}-\alpha({\bm l})\in \Omega_{{\omega_{1}}}}$$ is equivalent to $\boldsymbol{B_{U_{h}(g)}}$ by a diagonal transformation $Q$, i.e. $$\boldsymbol{B_{YY}}=Q  \boldsymbol{B_{U_{h}(g)}}Q^{-1}, Q\in End(V_{\omega_{1}}\otimes V_{\omega_{1}},\mathbb{Z}[t,t^{-1}]).$$
\end{theorem}

\begin{remark}
The advantage of considering thimbles as the basis for $\mathbb{Z}[t,t^{-1}]$-module $V_{\omega_{1}}\otimes V_{\omega_{1}}$ is that the imaginary part of holomorphic function is conserved along the thimble defined by the gradient flow of the real part of the function \cite{HL1}. When the parameter varies, there is a global variation for the thimble, which can be extracted from the critical point associated with it. Therefore, it is convenient to calculate the monodromy of $\mathfrak{J}\overset{\pi}{\rightarrow} X_{2}$ valued in $\mathbb{Z}[t,t^{-1}]$ only by considering the variation of the critical point of Yang-Yang function.
\end{remark}

Denote by $\mathfrak{J}_{0}\overset{\mu}{\rightarrow} X_{2}$ the thimble space generated by the family $$\{e^{-\frac{{\bm W}_{0}({\bm w},{\bm z},{\bm \lambda},{\bm l})}{\kappa+h^{\vee}}}\mid V_{2\omega_{1}-\alpha({\bm l})}  \subset  V_{\omega_{1}}\otimes V_{\omega_{1}}\}.$$ By theorem \ref{1sing}, it is  a fiber bundle of $\mathbb{Z}[t,t^{-1}]$-module with its fiber $\mathfrak{J}_{0}(P)$ isomorphic to $Sing V_{\omega_{1}}\otimes V_{\omega_{1}}$. Denote also by $\boldsymbol{B_{YY}}$ the monodromy representation on it induced by clockwise rotation $T(1)$ in equation \eqref{equ:T}.  In section \ref{sec:conclusion}, we will define transformation $\boldsymbol{S}:\mathfrak{J}_{0}(P)\rightarrow\mathfrak{J}(P)$ induced by the parameter deformations $c\rightarrow +\infty$ from $0$. Our second main theorem is as following:

\begin{theorem}\label{thm2} $\boldsymbol{S}$ commutes with $\boldsymbol{B_{YY}}$ i.e. the following diagram commutes.
$$\begin{array}[c]{ccc}\mathfrak{J}_{0}(P)&\stackrel{\boldsymbol{S}}{\rightarrow}&\mathfrak{J}(P)\\\downarrow\scriptstyle{\boldsymbol{B_{YY}}}&&\downarrow\scriptstyle{\boldsymbol{B_{YY}}}\\\mathfrak{J}_{0}(P)&\stackrel{\boldsymbol{S}}{\rightarrow}&\mathfrak{J}(P)\end{array}$$
\end{theorem}
\section{Wall-crossing formula and monodromy representation}\label{sec:formula}

 In this section, we study the monodromy representation $\boldsymbol{B_{YY}}$ of $\mathfrak{J}\overset{\pi}{\rightarrow} X_{2}$ and prove the first main theorem. To simplify the notation, we use $l$ to label each Yang-Yang function in the following sections. For $D_{n}$ Lie algebra, there are two Yang-Yang functions of ${\bm l}=(1,...,1,1,0)$ and ${\bm l}=(1,...,1,0,1)$ with the same $l=n-1$. We will label $(1,...,1,1,0)$ by $n-1$ and $(1,...,1,0,1)$ by $n-1'$ and define an order for them.

\subsection{A simple example of wall-crossing phenomena}\label{subsec:ex}

The key for the derivation of the monodromy representation is to figure out the wall-crossing formula. To illustrate it, we start from an example of $A_{1}$ Lie algebra $g=sl(2,\mathbb{C})$, $\dim V_{\omega_{1}}=2$ and $\Omega_{\omega_{1}}=\{\omega_{1},\omega_{1}-\alpha\}$. The inner product on the weight space is $(a,b)=\frac{a\cdot b}{2}$.
For $v_{1}\otimes v_{0},v_{0}\otimes v_{1}\in V_{\omega_{1}}\otimes V_{\omega_{1}}$, two Yang-Yang functions corresponding to them equal
\begin{equation}
\begin{split}
&{\bm W}_{c}({\bm w},{\bm z},{\bm \lambda},1+0)\\=&{\bm W}_{c}({\bm w},{\bm z},{\bm \lambda},0+1)\\
=&\sum_{a}(\alpha,\omega_{1})\log\left( w-z_{a}\right) -(\omega_{1},\omega_{1})\log \left( z_{1}-z_{2}\right)
-c(w-\frac{1}{2}\left( z_{1}+z_{2}\right)).
\end{split}
\end{equation}

Its critical point equation is
\begin{equation}
      \frac{1}{w-z_{1}}+ \frac{1}{w-z_{2}}=c,
\end{equation}

which has two solutions $w^{1}(c)$ and $w^{2}(c)$ for $c\geq 2$. Assume $$\lim_{c\rightarrow +\infty}w^{1}(c)=z_{1},\lim_{c\rightarrow +\infty}w^{2}(c)=z_{2}.$$
With large $c$, let $J_{1,0}$ and $J_{0,1}$ be two thimbles associated respectively to $w^{1}(c)$ and $w^{2}(c)$.
The continuous clockwise transformation $T(1)$ induces a continuous deformation on the thimble $J_{1,0}$:$$J_{1,0}\rightarrow q^{-\frac{1}{2}(\omega_{1}-\alpha,\omega_{1})}J_{0,1}=q^{\frac{1}{4}}J_{0,1},$$ where $q^{-\frac{1}{2}(\omega_{1}-\alpha,\omega_{1})}$ is from the phase factor difference of the critical values and it is equal to the phase factor difference of$ (z_{1}-z_{2})^{\frac{(\omega_{1}-\alpha,\omega_{1})}{\kappa+h^{\vee}}}$ under the $\frac{1}{2}$ clockwise rotation. The transformation of $J_{0,1}$ is more interesting. In the process of clockwise rotation $T(s)$, $z_{1}$ will pass through $J_{0,1}$ from the right hand side of $z_{2}$, when $s=\frac{1}{2}$, the imaginary parts of two critical values equals: $$Im {\bm W}_{c}(w^{1},T(\frac{1}{2})P,{\bm \lambda},1)=Im {\bm W}_{c}(w^{2},T(\frac{1}{2})P,{\bm \lambda},1).$$ There is a gradient flow connecting $w^{2}$ to $w^{1}$ as shown in figure 14 of \cite{GW}. The homotopic class of $J_{0,1}$ after deformation is equivalent to a zig-zag $\mathbb{Z}[t,t^{-1}]$-linear combination of $J_{1,0}$ and $J_{0,1}$：
\begin{equation}\label{equ:a1wc}
  \boldsymbol{B}J_{0,1}=aJ_{1,0}+(c-b)J_{0,1}.
\end{equation}
Because in general situation thimble is of high dimension, it is convenient to just draw the variation of its critical point along the homotopic class of the thimble after rotation, rather than to draw the thimble itself. Here, the zig-zag thimble has three parts and the relations of their critical points are shown in figure \ref{fig:t10}. In the following, whenever we draw the figure of the variation of the critical point, we are showing the relations between the critical points of the different thimbles in the homotopic class of the thimble after the deformation.

The coefficient $a$ in \eqref{equ:a1wc} is from the phase factor of  $e^{-\frac{{\bm W}_{c}(w^{2},T(\frac{1}{2})P,{\bm \lambda},1)}{\kappa+h^{\vee}}}\sim  (z_{1}-z_{2})^{\frac{(\omega_{1}-\alpha,\omega_{1})}{\kappa+h^{\vee}}}$. As shown in figure \ref{fig:t10}, $b$ differs from $a$ by an additional movement of the critical point from the right of $z_{2}$ to the right of $z_{1}$ and the negative sign before $b$ is from the orientation reverse of the thimble. $c$ differs from $b$ by an anti-clockwise rotation around $z_{1}$. Thus, $a=q^{-\frac{1}{2}(\omega_{1},\omega_{1}-\alpha)}$, $b=q^{\frac{1}{2}(\omega_{1},\alpha)}\cdot a$ and  $c=q^{-(\omega_{1},\alpha)}\cdot b$.

\begin{equation}\label{equ:a1wc2}
\boldsymbol{B}J_{0,1}=q^{\frac{1}{4}}J_{1,0}+(q^{-\frac{1}{4}}-q^{\frac{3}{4}})J_{0,1}.\\
\end{equation}
The minimal phase factor under $T$ is $t=q^{\frac{1}{4}}$ and therefore the monodromy is valued in $\mathbb{C}[q^{\frac{1}{4}},q^{-\frac{1}{4}}]$.
In sum,
\begin{equation}
  \boldsymbol{B}\left(
     \begin{array}{c}
       J_{1,0} \\
       J_{0,1} \\
     \end{array}
   \right)=\left(
             \begin{array}{cc}
               0 & q^{\frac{1}{4}} \\
               q^{\frac{1}{4}} & q^{-\frac{1}{4}}-q^{\frac{3}{4}} \\
             \end{array}
           \right)
\left(
     \begin{array}{c}
       J_{1,0} \\
       J_{0,1} \\
     \end{array}
   \right).
\end{equation}

\begin{figure}
 \centering\includegraphics[width=6cm]{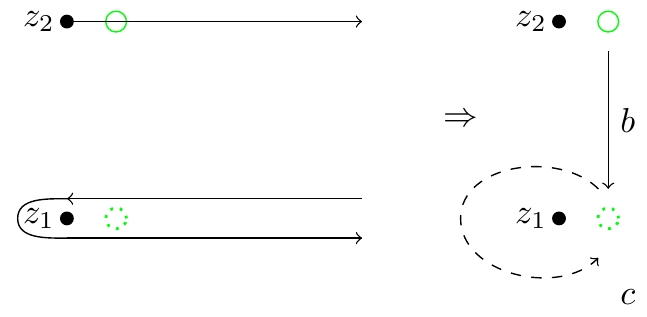}\\
\caption{Variation of the critical point on ${\bm W}$ plane.}
 \label{fig:t10}
 \end{figure}

The phenomena is called wall-crossing phenomena and the formula \eqref{equ:a1wc2} describing it is called wall-crossing formula. Terms with coefficients $b$ and $c$ are wall-crossing terms. The following properties are clear. Firstly, because critical points of the thimbles in the wall-crossing phenomena are different solutions from the same Yang-Yang function, thus the types and total number of primary roots will not be created or annihilated in the process  of wall-crossing, but only be transfered from one point to another point. We call this property conservation law of wall-crossing. Secondly, clockwise transformation $T$ make a constraint on the direction of the primary roots transfer: primary roots only move in the direction of positive real axis from $z_{2}$ to $z_{1}$. Based on the above two properties, $E^{\mathbbm{q}}$ is a $\sigma$ invariant sub-module. The monodromy representation $\boldsymbol{B}$ is naturally decomposed into a direct sum of the sub-representations on $E^{\mathbbm{q}}$ and all matrices of the sub-representation are triangular and valued in $\mathbb{Z}[t,t^{-1}]$.

\subsection{Variations of critical points, wall-crossing formula and monodromy representation}\label{subsct:monodromy}

To derive wall-crossing formula, it is necessary to analyze the variation of critical points under the transformation $T$. As in the previous example, we focus on the homotopic class of the thimble after deformation $T$ and see how the primary roots move from $z_{2}$ to $z_{1}$. For the fundamental representation $V_{\omega_{1}}$ of $g\in A_{n},B_{n},C_{n},D_{n}$, we conclude  as following four types of variations of critical points during wall-crossing and the details of derivation can be found in each case in section \ref{subsubsct:A}, \ref{subsubsct:B}, \ref{subsubsct:C} and \ref{subsubsct:D}.

\begin{description}

  \item[Type I]
\begin{figure}
 \centering\includegraphics[width=3.5cm]{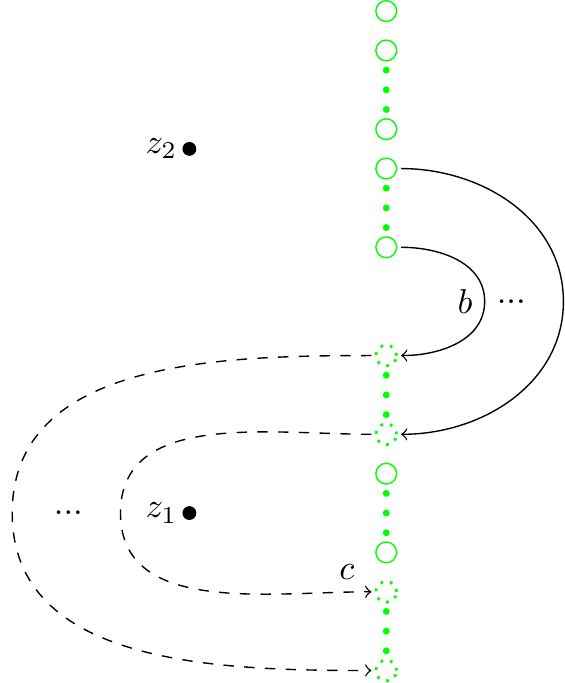}\\
\caption{Type I.}
 \label{fig:t1}
 \end{figure}

As shown in figure \ref{fig:t1}, coordinates of the critical point have the same real parts but different imaginary parts. This variation appears in the case of the finite dimensional irreducible representation of $A_{1}$ or in the case of $B_{n}$ Lie algebra.
  \item[Type II]

\begin{figure}
 \centering\includegraphics[width=3.5cm]{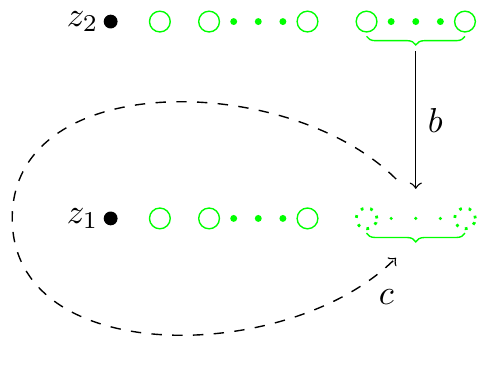}\\
\caption{Type II.}
 \label{fig:t2}
 \end{figure}
Coordinates of the critical point near $z_{2}$ have the same imaginary parts formulated as a straight line paralleled to the real axis. The variation is just translation as shown in figure \ref{fig:t2}. The homotopic class of the thimble after deformation is equivalent to three parts. There is only one dimension in  the sub-thimble reversing its orientation, thus the sign before the coefficient $b$ is always minus in this case.
  \item[Type III]
\begin{figure}
 \centering\includegraphics[width=4cm]{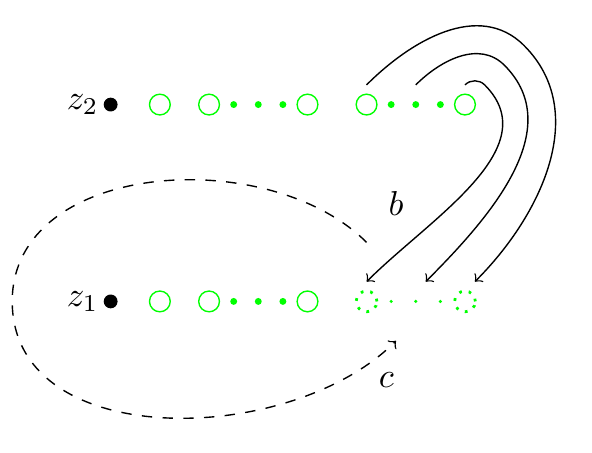}\\
\caption{Type III.}
 \label{fig:t3}
 \end{figure}

Coordinates of the critical point near $z_{2}$ have the same imaginary parts formulated as a straight line paralleled to the real axis. The variation from $z_{2}$ to $z_{1}$ as shown in figure \ref{fig:t3} has a clockwise self rotation of $\pi$, then followed by a translation. In this case, the orientation of each dimension of the sub-thimble connecting $z_{2}$ to $z_{1}$ is reversed. Thus, the sign before $b$ is $(-1)^{j}$, where $j$ is the dimension of the sub-thimble or the number of primary roots moving from $z_{2}$ to $z_{1}$. Opposite to the sign before $b$, it is $(-1)^{j+1}$ before $c$ .
  \item[Type IV]

\begin{figure}
 \centering\includegraphics[width=4cm]{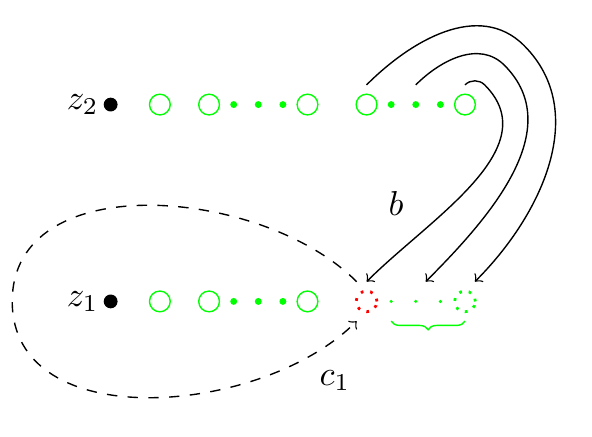} \centering\includegraphics[width=4cm]{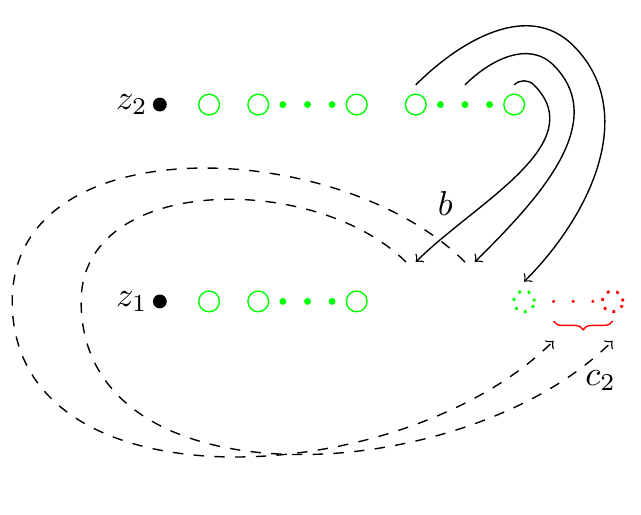}
\centering\includegraphics[width=4cm]{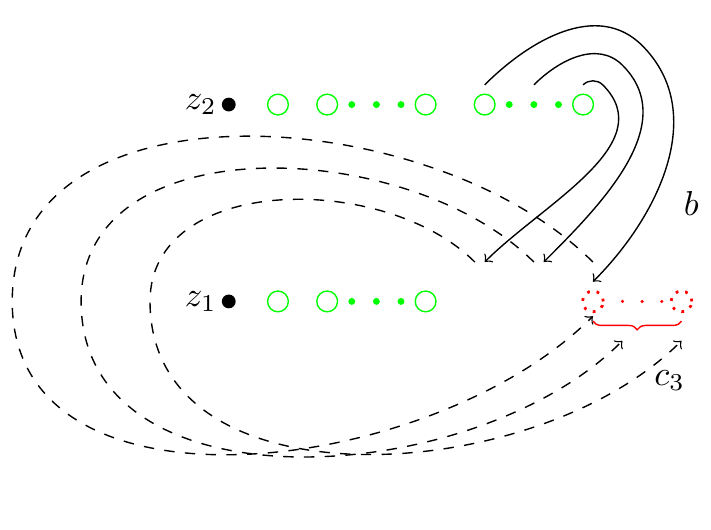}\\
\caption{The origin of $c_{1}$, $c_{2}$ and $c_{3}$ of type IV.}
 \label{fig:t4}
 \end{figure}

Coordinates of the critical point near $z_{2}$ have the same imaginary parts formulated as a straight line paralleled to the real axis. The variation leaves a trace along the homotopic class of the integration cycle like a "snake" keeping the relative order of the coordinates invariant during the moving. The coefficient $b$ is from the variation of primary roots moving from $z_{2}$ to $z_{1}$. $c_{1}$, $c_{2}$ and $c_{3}$ differs from $b$ respectively by $1$, $j-1$ and $j$ primary roots rotating around $z_{1}$ and other primary roots near $z_{1}$ in a specific manner, as is shown in the pictures of figure \ref{fig:t4}. Thus, the sign before $b$, $c_{1}$, $c_{2}$ and $c_{3}$ are $(-1)^{j}$, $(-1)^{j+1}$, $(-1)^{2j-1}=-1$ and $(-1)^{2j}=1$ respectively.
\end{description}

In the following, we derive the monodromy representations for $A_{n}$, $B_{n}$, $C_{n}$ and $D_{n}$ respectively.

\subsubsection{$A_{n}$}\label{subsubsct:A}
Denote by $\{\lambda^{i}\}_{i=0,1,...,n}$ the weights of $V_{\omega_{1}}$, where $\lambda^{i}=\omega_{1}-\sum_{j=1}^{i}\alpha_{j}$. For each weight vector $v_{\lambda^{j}}\in V_{\omega_{1}}$, the corresponding Yang-Yang function${\bm W}_{c}({\bm w},z,\omega_{1},j)$ with one parameter $z$ is
\begin{equation}
\begin{split}\label{equ:1zGYY in SB}
{\bm W}_{c}({\bm w},z,\omega_{1},j)
=\sum _{i=1}^{j}(\alpha_{i},\omega_{1})\log\left( w_{i}-z\right) &-\sum_{1\leq i< s\leq j}(\alpha_{i},\alpha_{s})\log \left( w_{i}-w_{s}\right)\\
&-c(\sum_{i=1}^{j}(\alpha_{i},\rho)w_{i}- (\omega_{1},\rho)z).
\end{split}
\end{equation}
Since $(\omega_{1},\alpha_{i})=\delta_{i,1}$ and $$(\alpha_{i},\alpha_{k})=\left\{
                                                                   \begin{array}{ll}
                                                                     2, & \hbox{$i=k$;} \\
                                                                     -1, &\hbox{$\mid i-k\mid=1$;}\\
                                                                   0, & \hbox{otherwise,}
                                                                   \end{array}
                                                                 \right.$$
\begin{equation}
\begin{split}
{\bm W}_{c}({\bm w},z,\omega_{1},j)
=\log\left( w_{1}-z\right) &+\sum_{i=1}^{j-1}\log \left( w_{i}-w_{i+1}\right)-c(\sum_{i=1}^{j}w_{i}- \frac{n}{2}z).
\end{split}
\end{equation}
Its critical point equation is as following:
\begin{equation}\label{ceA}
  \left\{
    \begin{aligned}
      \frac{1}{w_{1}-z}&=\frac{-1}{w_{1}-w_{2}}+c  \\
      0&=\frac{-1}{w_{2}-w_{1}}+\frac{-1}{w_{2}-w_{3}}+c  \\
      ...   \\
      0&=\frac{-1}{w_{j-1}-w_{j-2}}+\frac{-1}{w_{j-1}-w_{j}}+c \\
      0&=\frac{-1}{w_{j}-w_{j-1}}+c.
    \end{aligned}
  \right.
\end{equation}

It is obvious that $z$ and the solution $\{w_{i}=z+\sum_{k=0}^{i-1}\frac{1}{(j-k)c}\}_{i=1,..,j}$ are on the same horizontal line of ${\bm W}$ plane.

\begin{figure}
\centering\includegraphics[width=7.5cm]{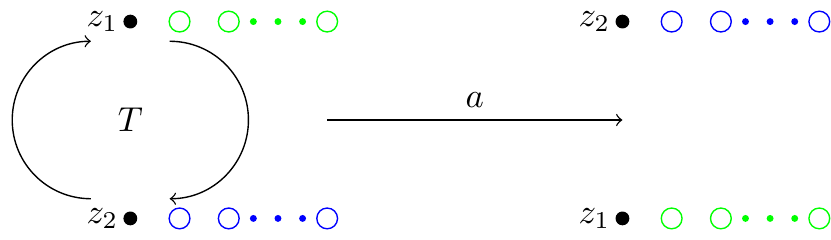}\\
\caption{Coefficient $a$ from basic clockwise rotation.}
 \label{fig:t0}
 \end{figure}

Therefore, for $v_{\lambda^{i}}\otimes v_{\lambda^{j}}\in V_{\omega_{1}}\otimes V_{\omega_{1}}$, the coordinates of the critical point are distributed respectively along two horizontal straight lines started from $z_{1}$ and $z_{2}$ on the ${\bm W}$ plane, as shown in the first picture of figure \ref{fig:t0}.

When $i\geq j$, there is no wall-crossing under the transformation $T$. $$\boldsymbol{B}J_{i,j}=q^{-\frac{1}{2}(\lambda^{i},\lambda^{j})}J_{j,i}.$$

When $i<j$, the variation is moving primary roots $\{\alpha_{i+1},\alpha_{i+2},...,\alpha_{j}\}$ from $z_{2}$ to $z_{1}$. Let $w^{a}_{k}, a=1,2$ be the coordinates of $\alpha_{k}$. The type of variation can be seen from the deformation of parameter $c\rightarrow+\infty$ from $c=0$. The critical point equation with two singularities $w^{1}_{i}$ and $w^{2}_{i}$ is as following:
\begin{equation}\label{equ:ANSB}
  \left\{
    \begin{aligned}
      \frac{1}{w_{i+1}-w^{1}_{i}}+\frac{1}{w_{i+1}-w^{2}_{i}}&=\frac{-1}{w_{i+1}-w_{i+2}}\\
      ...   \\
      0&=\frac{-1}{w_{k}-w_{k-1}}+\frac{-1}{w_{k}-w_{k+1}} \\
      ...  \\
      0&=\frac{-1}{w_{j}-w_{j-1}}+c,
    \end{aligned}
  \right.
\end{equation} where $i+2\leq k\leq j-1$.

For $c\neq 0$,  $$\begin{aligned}
w_{i+1}&=\frac{2 + c w^{1}_{i} + c w^{2}_{i} \pm \sqrt{4 + c^{2} (w^{1}_{i}-w^{2}_{i})^{2}}}{2c},\\
w_{k}&=w_{i+1}+\frac{k-i-1}{c},\quad k\geq i+2.\\
    \end{aligned}$$

$$\lim_{c\rightarrow 0}w_{i+1}=\frac{w^{1}_{i} +  w^{2}_{i}}{2}\quad \hbox{and}\quad \lim_{c\rightarrow 0}w_{k}=+\infty, \quad i+2\leq k\leq j.$$

$$\lim_{c\rightarrow +\infty}w_{i+1}=w^{1}_{i} \quad \hbox{or} \quad  w^{2}_{i}.$$

\begin{figure}
 \centering\includegraphics[width=7cm]{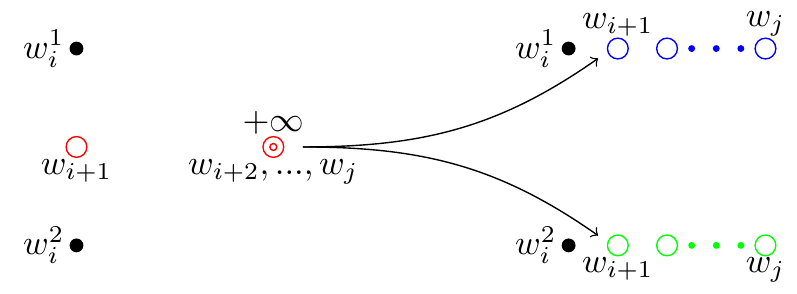}\\
\caption{Variation of critical points as $c\rightarrow +\infty$ from $0$.}
 \label{fig:va}
 \end{figure}
Note that the process above is independent of the position of $w^{1}_{i}$ and $w^{2}_{i}$. As shown in figure \ref{fig:va}, when $c\rightarrow +\infty$, the coordinates of the critical point at $c=0$ are continuously moving to $w^{1}_{i}$ or $w^{2}_{i}$ within the same horizontal line on ${\bm W}$ plane. In the process of continuous transformation $T(s)$, when $s=1/2$, the imaginary parts of two critical values equals. Thus there will be two thimbles connecting $w^{2}_{i}$ to $w^{1}_{i}$ from above and below $z_{1}$ and they are homotopically equivalent to the thimbles at $c=0$ with $Im w^{1}_{i}<Im w^{2}_{i}$ and $Im w^{1}_{i}>Im w^{2}_{i}$ respectively. The variation of $\{w_{i+1},...,w_{j}\}$ is just a translation and thus a type $II$ variation. As shown in figure \ref{fig:a1}, the homotopic class of $\boldsymbol{B}J_{i,j}$ is equivalent to  $\mathbb{Z}[t,t^{-1}]$-linear combination of three parts $$\boldsymbol{B}J_{i,j}=aJ_{j,i}+(c-b)J_{i,j}$$ and the difference of them coming from the translation of the primary roots $\{\alpha_{i+1},\alpha_{i+2},...,\alpha_{j}\}$.
In fact, under the basic clockwise rotation as shown in figure \ref{fig:t0}, $a$ is from the phase factor difference of $e^{-\frac{{\bm W}_{c}({\bm w},{\bm z},{\bm \lambda},i+j)}{k+h^{v}}}$ i.e. that of $(z_{1}-z_{2})^{(\lambda^{i},\lambda^{j})}$. Comparing with $a$, $b$ has an additional phase factor $q^{\frac{1}{2}(\lambda^{i},\lambda^{i}-\lambda^{j})}$ of translating  $\{\alpha_{i+1},\alpha_{i+2},...,\alpha_{j}\}$ from $z_{2}$ to $z_{1}$. $c$ differs from $b$ by translation of $\{\alpha_{i+1},\alpha_{i+2},...,\alpha_{j}\}$ around $z_{1}$ and $\{\alpha_{1},\alpha_{2},...,\alpha_{i}\}$. Thus$$c=b\cdot q^{-(\lambda^{i},\lambda^{i}-\lambda^{j})}.$$ Therefore, we get the following wall-crossing formula for $V_{\omega_{1}}$:
$$\boldsymbol{B}J_{i,j}=\left\{
             \begin{array}{ll}
               q^{-\frac{1}{2}(\lambda^{i},\lambda^{j})}J_{j,i}, & \hbox{$i\geq j$;} \\
               q^{-\frac{1}{2}(\lambda^{i},\lambda^{j})}(J_{j,i}+q^{\frac{1}{2}(\lambda^{i},\lambda^{i}-\lambda^{j})}(q^{-(\lambda^{i},\lambda^{i}-\lambda^{j})}-1)J_{i,j}), & \hbox{$j>i$.}
             \end{array}
           \right.
$$

\begin{figure}
 \centering\includegraphics[width=4.5cm]{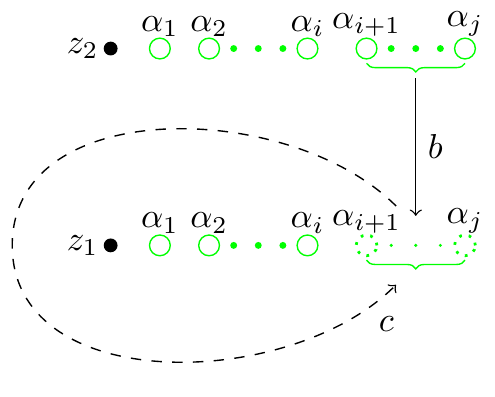}\\
\caption{In the fundamental representation of $A_{n}$ Lie algebra, wall-crossing coefficient $b$ is different from $a$ by the translation of $\{\alpha_{i+1},\alpha_{i+2},...,\alpha_{j}\}$ from $z_{2}$ to $z_{1}$ and $c$ different from $b$ by anti-clockwise translation around $z_{1}$ and $\{\alpha_{1},\alpha_{2},...,\alpha_{i}\}$.}
 \label{fig:a1}
 \end{figure}

\subsubsection{$B_{n}$}\label{subsubsct:B}
Denote the weights of the fundamental representation of $B_{n}$ Lie algebra by $$\lambda^{i}=\left\{
                                                                                        \begin{array}{ll}
                                                                                          \omega_{1}-\sum_{j=1}^{i}\alpha_{j}, & \hbox{$i\leq n$;} \\
                                                                                          \omega_{1}-\sum_{j=1}^{n}\alpha_{j}-\sum_{j=2n+1-i}^{n}\alpha_{j}, & \hbox{$n<i\leq 2n$.}
                                                                                        \end{array}
                                                                                      \right.
$$
Their inner products are:
\begin{equation}
(\lambda^{s},\lambda^{t})=\left\{
                            \begin{array}{ll}
                              1, & \hbox{$s+t\neq2n,s=t$;} \\
                              0, & \hbox{$s+t\neq2n,s\neq t$;} \\
                              0, & \hbox{$s+t=2n,s=t $;} \\
                              -1, & \hbox{$s+t=2n,s\neq t $,}
                            \end{array}
                          \right.
\end{equation}where $s,t=0,1,..,2n$.
For any weight vector $v_{\lambda^{l}}\in V_{\omega_{1}}$, the corresponding Yang-Yang function is as following:
\begin{equation}
\begin{split}\label{equ:BGYY in SB}
{\bm W}_{c}({\bm w},z,\omega_{1},l)
=\sum _{j=1}^{l}(\alpha_{i_{j}},\omega_{1})\log\left( w_{j}-z\right) -&\sum_{1\leq j< k\leq l}(\alpha_{i_{j}},\alpha_{i_{k}})\log \left( w_{j}-w_{k}\right)\\
-&c(\sum_{j=1}^{l}(\alpha_{i_{j}},\rho)w_{j}-(\omega_{1},\rho) z),
\end{split}
\end{equation}where $$i_{j}=\left\{
                                      \begin{array}{ll}
                                        j, & \hbox{$j\leq n$;} \\
                                        2n-j+1, & \hbox{$n<j\leq 2n$.}
                                      \end{array}
                                    \right.
$$
By this notation, when $l\geq n+1$, $\{w_{k}, w_{2n+1-k}\}_{2n+1-l\leq k\leq n}$ are pairs of symmetric coordinates of $\alpha_{k}$ in the function. Define $\bar{w}_{k}=w_{k}+w_{2n+1-k}$ and $\Delta_{k}=(w_{k}-w_{2n+1-k})^{2}$, $2n+1-l\leq k\leq n$.

\begin{lemma}\label{lem:exsol} For the fundamental representation $V_{\omega_{1}}$ of $B_{n}$ Lie algebra, the solutions of the critical point equation \eqref{CE2} of the corresponding Yang-Yang functions ${\bm W}_{c}({\bm w},0,\omega_{1},l)$ are as following:

When $l<n$, \begin{equation}
w_{j}=\sum_{i=1}^{j}\frac{1}{c(l-i+1)}\quad j=1,...,l.\end{equation}

When $l=n$, \begin{equation}w_{j}=\sum_{i=1}^{j}\frac{1}{c(l-i+1/2)}\quad j=1,...,l.\end{equation}

When $l\geq n+1$, \begin{equation}w_{k}=\sum_{i=1}^{k}\frac{1}{c(l-i)}\end{equation} for $k=1,...,2n-l.$
For $2n+1-l\leq k\leq n$, \begin{equation}\begin{split}\bar{w}_{k}=&\frac{1}{c(l-r-1)}+\sum_{j=1}^{2n-l}\frac{2}{c(l-j)}+\sum_{j=2n-l+1}^{k-1}\frac{1}{c(l-j-1)}\\&+\sum_{j=2n-l+1}^{2n-k-1}\frac{1}{c(l-j-1)},\\
\Delta_{k}=&[\sum_{j=k}^{2n-k-1}\frac{1}{c(l-j-1)}]^2-\frac{1}{c^2(l-n-1)^2}.\end{split}\end{equation}
\end{lemma}
\begin{proof} The explicit solutions of the critical point equation of Yang-Yang function ${\bm W}({\bm w},{\bm z},{\bm \lambda},{\bm l})$ with ${\bm z}=(0,1)$, ${\bm \lambda}=(\lambda,\omega_{1})$ are already known in \cite{LMV16}. If we denote them by $w_{1}(\lambda),...,w_{l}(\lambda)$, then the solutions $\tilde{w}_{1},...,\tilde{w}_{l}$ of the critical point equation \eqref{CE2} of Yang-Yang functions ${\bm W}_{c}({\bm w},0,\omega_{1},l)$ are just the limitation of the critical solutions of ${\bm W}({\bm w},{\bm z},{\bm \lambda},{\bm l})$ with the data ${\bm z}=(0,\mu)$, ${\bm \lambda}=(\mu c\rho,\omega_{1})$:
$$\tilde{w}_{j}=\lim_{\mu\rightarrow +\infty}\mu(1-w_{j}(\mu c\rho)),\quad j=1,...l.$$
By the limitation, the explicit solutions of equation \eqref{CE2} are straightforward.
\end{proof}

By this lemma, the following property is clear.

\begin{lemma} $\Delta_{k}>0$ for $k=2n+1-l,...,n-1$ and $\Delta_{n}<0$. The coordinates of critical solutions satisfy the following order:

When $l\leq n$, \begin{equation}0<w_{1}<w_{2}<...<w_{l};\end{equation}
When $l\geq n+1$, assume $w_{k}<w_{2n+1-k}$ for $k=2n+1-l,...,n-1$, then
\begin{equation}0<w_{1}<...<w_{2n-l}<w_{2n+1-l}<...<w_{n-1}<\frac{\bar{w}_{n}}{2}<w_{n+2}<...<w_{l}.\end{equation}
\end{lemma}
If the critical points of ${\bm W}_{c}({\bm w},0,\omega_{1},l)$ are ${\bm w}$, then the critical points of ${\bm W}_{c}({\bm w},z,\omega_{1},l)$ are $z+{\bm w}$. By the previous lemma, it is obvious that $z$ and $\{w_{k}\}_{k\neq n,n+1}$ are on the same horizontal line of ${\bm W}_{c}$ plane except $w_{n}$ and $w_{n+1}$ vertically symmetrical about the center point $\frac{\bar{w}_{n}}{2}$. In sum, we draw the distribution of the critical point near $z$ corresponding to $v_{\lambda^{i}}\in V_{\omega_{1}}$ in different cases in figure \ref{fig:db}.

\begin{figure}
 \centering\includegraphics[width=2.5cm]{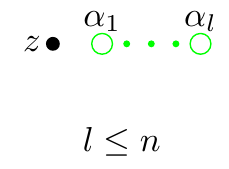}
\centering\includegraphics[width=3cm]{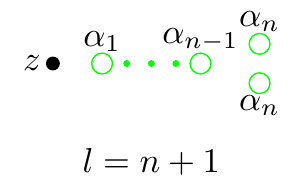}\\
\centering\includegraphics[width=6cm]{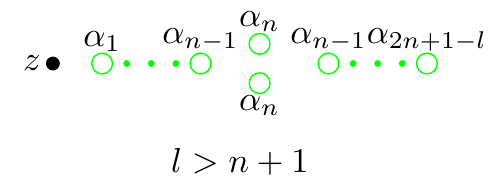}\\
\caption{Coordinates distribution of the $B_{n}$ critical point on ${\bm W}_{c}$ plane near $z$.}
 \label{fig:db}
 \end{figure}

In the following ,we derive monodromy representation for $v_{\lambda^{i}}\otimes v_{\lambda^{j}}\in V_{\omega_{1}}\otimes V_{\omega_{1}}$. it is convenient to consider the case $i+j\neq 2n$ firstly.

i)$i+j\neq 2n\& i\geq j $

There is no wall-crossing, $$\boldsymbol{B}J_{i,j}=q^{-\frac{1}{2}(\lambda^{i},\lambda^{i})}J_{j,i}.$$

ii)$i+j\neq 2n\& i<j \& i\neq n$

By the same method used in the case of $A_{n}$, the variation is of type $II$.
\begin{equation}
\begin{split}
  &\boldsymbol{B}J_{i,j}\\
=&q^{-\frac{1}{2}(\lambda^{i},\lambda^{i})}(J_{j,i}+q^{\frac{1}{2}(\lambda^{i},\lambda^{i}-\lambda^{j})}(q^{-(\lambda^{i},\lambda^{i}-\lambda^{j})}-1)J_{i,j})\\
=&J_{j,i}+(q^{-\frac{1}{2}}-q^{\frac{1}{2}})J_{i,j}.\\
\end{split}
\end{equation}

\begin{figure}
\centering
\centering\includegraphics[width=3cm]{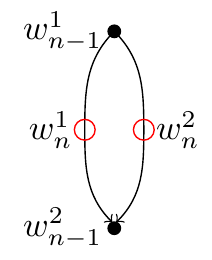}\\
\caption{Thimble with $c=0$ and its critical point $w^{1}_{n}$ and $w^{2}_{n}$.}
 \label{fig:vanf}
 \end{figure}

\begin{figure}
\centering\includegraphics[width=4.5cm]{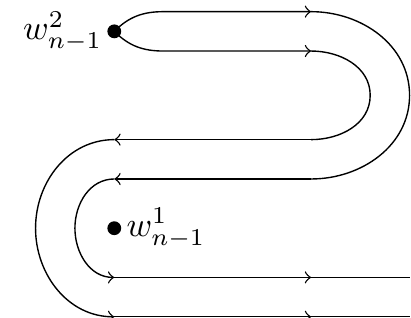}\\
\caption{Homotopic class of the thimble after rotation.}
 \label{fig:vanf2}
 \end{figure}

\begin{figure}
\centering\includegraphics[width=4.5cm]{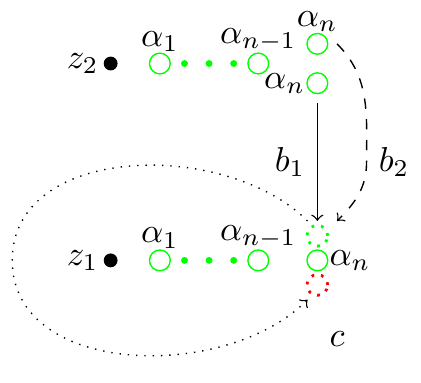}\\
\caption{Type $I$ wall-crossing of $B_{n}$.}
 \label{fig:dbt1}
 \end{figure}

\begin{figure}
\centering\includegraphics[width=5.5cm]{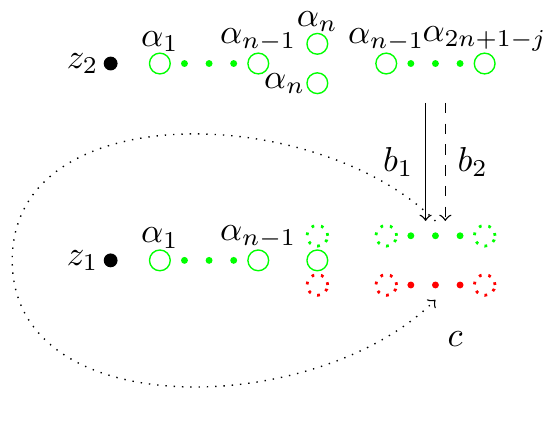}\\
\caption{Combination of type $I$ and $II$.}
 \label{fig:dbt12}
 \end{figure}

iii)$i+j\neq 2n\& i=n\&j=n+1 $

\begin{lemma}\label{rlemma} Let $\bar{w}_{i}=w^{1}_{i}+w^{2}_{i}$, $i=1,2$. Assume that $w^{1}_{1}\neq w^{2}_{1}$, $$\frac{2w^{1}_{1}-\bar{w}_{2}}{(w^{1}_{1})^{2}-\bar{w}_{2}w^{1}_{1}+w^{1}_{2}w^{2}_{2}}=A,\quad\frac{2w^{2}_{1}-\bar{w}_{2}}{(w^{2}_{1})^{2}-\bar{w}_{2}w^{2}_{1}+w^{1}_{2}w^{2}_{2}}=-A$$
and $$A(w^{1}_{1}-w^{2}_{1})\neq 2.$$

Then $$\bar{w}_{2}=\bar{w}_{1},\quad w^{1}_{2}w^{2}_{2}=w^{1}_{1}w^{2}_{1}+A^{-1}(w^{1}_{1}-w^{2}_{1}).$$
\end{lemma}

\begin{proof}  The equations imply $$A(w^{1}_{1}-w^{2}_{1})(\bar{w}_{1}-\bar{w}_{2})=2(\bar{w}_{1}-\bar{w}_{2}).$$
Then $A(w^{1}_{1}-w^{2}_{1})\neq 2$ implies $$\bar{w}_{2}=\bar{w}_{1},\quad w^{1}_{2}w^{2}_{2}=w^{1}_{1}w^{2}_{1}+A^{-1}(w^{1}_{1}-w^{2}_{1}).$$
\end{proof}

By this lemma, the following critical point equation with $c=0$ can be solved:
\begin{equation}\label{tsceB1}
  \left\{
    \begin{aligned}
      \frac{1}{w^{1}_{n}-w^{1}_{n-1}}+\frac{1}{w^{1}_{n}-w^{2}_{n-1}}&=\frac{1}{w^{1}_{n}-w^{2}_{n}}\\
      \frac{1}{w^{2}_{n}-w^{1}_{n-1}}+\frac{1}{w^{2}_{n}-w^{2}_{n-1}}&=\frac{1}{w^{2}_{n}-w^{1}_{n}},
    \end{aligned}
  \right.
\end{equation} where $w^{1}_{n},w^{2}_{n}$ are coordinates of $\alpha_{n}$ and $w^{1}_{n-1},w^{2}_{n-1}$ coordinates of $\alpha_{n-1}$.
The solution is $\bar{w}_{n}=\bar{w}_{n-1},\quad
w^{1}_{n}w^{2}_{n}=\frac{(\bar{w}_{n-1})^{2}-w^{1}_{n-1}w^{2}_{n-1}}{3}.
$ If $Re w^{1}_{n-1}=Re w^{2}_{n-1}$, then $\Delta_{n}>0$. The thimble and its critical point are shown in figure \ref{fig:vanf}. When $c\rightarrow +\infty$, the condition $i=n\&j=n+1$ corresponds to the case of $w^{1}_{n},w^{2}_{n}$ tending to $w^{2}_{n-1}$. After rotation of $w^{1}_{n-1}$ and $w^{2}_{n-1}$, the homotopic class of the thimble connecting $w^{2}_{n-1}$ with $\infty$ is shown in figure \ref{fig:vanf2}.
Therefore, as is shown in figure \ref{fig:dbt1}, the variation here is of type $I$ with only one primary root $\alpha_{n}$ moving.

\begin{equation}
\boldsymbol{B}J_{i,j}
=aJ_{j,i}+(c_{1}+c_{2}-b_{1}-b_{2})J_{i,j},
\end{equation}
where $$b_{1}+b_{2}=a\cdot q^{\frac{1}{2}(\lambda^{n-1},\alpha_{n})}(1+q^{-\frac{1}{2}(\alpha_{n},\alpha_{n})})$$ is the sum of phase factors of moving one of two $\alpha_{n}$ to the right of $z_{1}$ and $$c_{1}+c_{2}=(b_{1}+b_{2})\cdot q^{-(\lambda^{n-1},\alpha_{n})+\frac{1}{2}(\alpha_{n},\alpha_{n})}$$
from anti-clockwise rotation of $\alpha_{n}$, $2\pi$ around $\{z_{1}, \alpha_{1}, ... , \alpha_{n-1}\}$ and $\pi$ around $\alpha_{n}$. Minus before $b_{1}$ and $b_{2}$ is from the reverse of the direction of the dimension one sub-thimble connecting $z_{2}$ to $z_{1}$.
Thus, \begin{equation}
\boldsymbol{B}J_{i,j}
=J_{j,i}+(q^{-\frac{1}{2}}-q^{\frac{1}{2}})J_{i,j}.
\end{equation}

iv)$i+j\neq 2n\& i=n\&j>n+1 $

The variation is combination of type $I$ and $II$. As shown in figure \ref{fig:dbt12}, similar to the previous case, there are two possible way of moving $\alpha_{n}$ from $z_{2}$ to $z_{1}$, but $\alpha_{n}$ is now accompanied by $\{\alpha_{n-1}, ..., \alpha_{2n+1-j}\}$ horizontally.
\begin{equation}
\begin{split}
 a&=q^{-\frac{1}{2}(\lambda^{i},\lambda^{i})};\\
b_{1}+b_{2}&=a\cdot q^{\frac{1}{2}(\lambda^{n-1},\alpha_{n})+\frac{1}{2}(\lambda^{n},\lambda^{n+1}-\lambda^{j})}(1+q^{-\frac{1}{2}(\alpha_{n},\alpha_{n})});\\
c_{1}+c_{2}&=(b_{1}+b_{2})\cdot q^{-(\lambda^{n-1},\alpha_{n})+\frac{1}{2}(\alpha_{n},\alpha_{n})-(\lambda^{n},\lambda^{n+1}-\lambda_{j})};\\
\boldsymbol{B}J_{i,j}&=aJ_{j,i}+(c_{1}+c_{2}-b_{1}-b_{2})J_{i,j}\\
&=J_{j,i}+(q^{-\frac{1}{2}}-q^{\frac{1}{2}})J_{i,j}.\\
\end{split}
\end{equation}
Minus is always from reverse of orientation of the dimension one sub-thimble.

In sum, when $i+j\neq 2n$,
\begin{equation}
  \boldsymbol{B}J_{i,j}=\left\{
             \begin{array}{ll}
               q^{-\frac{1}{2}}J_{j,i}, & \hbox{$i+j\neq 2n\&i=j$;} \\
               J_{j,i}, & \hbox{$i+j\neq 2n\&i> j$;} \\
               J_{j,i}+(q^{-\frac{1}{2}}-q^{\frac{1}{2}})J_{i,j}, & \hbox{$i+j\neq 2n\&i<j$.}
             \end{array}
           \right.
\end{equation}

When $i+j=2n$, the situation is more interesting.
The critical point equation of two singularities $z_{1}$ and $z_{2}$ with $c=0$ is as following:
\begin{equation}\label{equ:BNSB}
  \left\{
    \begin{aligned}
      0&=\frac{-1}{w^{1}_{1}-z_{1}}+\frac{-1}{w^{1}_{1}-z_{2}}+\frac{2}{w^{1}_{1}-w^{2}_{1}}+\frac{-1}{w^{1}_{1}-w^{1}_{2}}+\frac{-1}{w^{1}_{1}-w^{2}_{2}}\\
      0&=\frac{-1}{w^{2}_{1}-z_{1}}+\frac{-1}{w^{2}_{1}-z_{2}}+\frac{2}{w^{2}_{1}-w^{1}_{1}}+\frac{-1}{w^{2}_{1}-w^{1}_{2}}+\frac{-1}{w^{2}_{1}-w^{2}_{2}}\\
      ...   \\
      0&=\frac{2}{w^{1}_{k}-w^{2}_{k}}+\frac{-1}{w^{1}_{k}-w^{1}_{k-1}}+\frac{-1}{w^{1}_{k}-w^{2}_{k-1}}+\frac{-1}{w^{1}_{k}-w^{1}_{k+1}}+\frac{-1}{w^{1}_{k}-w^{2}_{k+1}}\\
      0&=\frac{2}{w^{2}_{k}-w^{1}_{k}}+\frac{-1}{w^{2}_{k}-w^{1}_{k-1}}+\frac{-1}{w^{2}_{k}-w^{2}_{k-1}}+\frac{-1}{w^{2}_{k}-w^{1}_{k+1}}+\frac{-1}{w^{2}_{k}-w^{2}_{k+1}}\\
       ...  \\
      0&=\frac{1}{w^{1}_{n}-w^{2}_{n}}+\frac{-1}{w^{1}_{n}-w^{1}_{n-1}}+\frac{-1}{w^{1}_{n}-w^{2}_{n-1}}\\
      0&=\frac{1}{w^{2}_{n}-w^{1}_{n}}+\frac{-1}{w^{2}_{n}-w^{1}_{n-1}}+\frac{-1}{w^{2}_{n}-w^{2}_{n-1}},\\
    \end{aligned}
  \right.
\end{equation}
where $2\leq k\leq n-1$ and $w^{1}_{i},w^{2}_{i}$ are coordinates of $\alpha_{i}$, $1\leq i\leq n$.
Summing up all the equations except the first pair gives   $$\frac{1}{w^{1}_{1}-w^{1}_{2}}+\frac{1}{w^{1}_{1}-w^{2}_{2}}+\frac{1}{w^{2}_{1}-w^{1}_{2}}+\frac{1}{w^{2}_{1}-w^{2}_{2}}=0.$$ Let $A=+\frac{2}{w^{1}_{1}-w^{2}_{1}}+\frac{-1}{w^{1}_{1}-w^{1}_{2}}+\frac{-1}{w^{1}_{1}-w^{2}_{2}}$, then $\frac{2}{w^{2}_{1}-w^{1}_{1}}+\frac{-1}{w^{2}_{1}-w^{1}_{2}}+\frac{-1}{w^{2}_{1}-w^{2}_{2}}=-A$. By using lemma \ref{rlemma} inductively, $\bar{w}_{1}=...=\bar{w}_{n}=z_{1}+z_{2}$. Then substituting $\bar{w}_{l}$ into \eqref{equ:BNSB}, we have the following solution: $$\bar{w}_{1}=...=\bar{w}_{n}=z_{1}+z_{2},\quad w^{1}_{l}w^{2}_{l}=z_{1}z_{2}+\frac{(z_{1}-z_{2})^{2}l(2n-l)}{4n^{2}-1},1\leq l\leq n.$$
It is clear that $\Delta_{k}<0$, $k=1,...,n-1$ and $\Delta_{n}>0$.
The distribution of the coordinates of critical point is shown in figure \ref{fig:nsbb}. Similar to the case of one singularity in figure \ref{fig:db},
the coordinates are symmetric on the same line connecting $z_{1}$ and $z_{2}$ except $w^{1}_{n}$ and $w^{2}_{n}$ vertically in the middle. Assume $Im w^{1}_{i}>Im w^{2}_{i},\quad i=1,...,n-1$. The different imaginary parts of $\{w^{1}_{i},w^{2}_{i}\}$ and $z_{1}, z_{2}$  give the coordinates a partial order: $$z_{1}<w^{1}_{1}<...<w^{1}_{n-1}<w^{1}_{n},w^{2}_{n}<w^{2}_{n-1}<...<w^{2}_{1}<z_{2}.$$
Because of the existence of the thimble above, when $c\rightarrow +\infty$, there is a wall-crossing for every $i\neq 0$ during the transformation $T$ and its homotopic class will keep this order.

To be precise, we redefine the index $i$ and $j$. Let $i$ be the number of primary roots near $z_{2}$ and $j$ be the number of primary roots crossing wall. Assume that
\begin{equation}\label{equ:cof}
\boldsymbol{B}J_{a,b}=\sum_{c,d}\boldsymbol{B}^{c,d}_{a,b}J_{c,d},
\end{equation}
 then from the properties of wall-crossing,
\begin{equation}\label{equ:Bwc}
\boldsymbol{B}J_{2n-i,i}=\sum_{j=0}^{i}\boldsymbol{B}^{i-j,2n-i+j}_{2n-i,i}J_{i-j,2n-i+j}.
\end{equation}
Phase factor from the rotation gives $$\boldsymbol{B}^{i,2n-i}_{2n-i,i}=q^{-\frac{1}{2}(\lambda^{i},\lambda^{2n-i})}.$$
The coefficients $\boldsymbol{B}_{2n-i,i}^{i-j,2n-i+j} (j>0)$ of wall-crossing formula can be computed as in the following cases.

i) $0<j\leq i\leq n-1$

Keeping the order of the coordinates, moving of $\{\alpha_{i-j+1}, ..., \alpha_{i}\}$ from $z_{2}$ to $z_{1}$ is accompanied by the self clockwise rotation of $\pi$, thus the variation is of type $III$.
\begin{equation}
  \boldsymbol{B}J_{2n-i,i}
=aJ_{i,2n-i}+((-1)^{j}b+(-1)^{j+1}c)J_{i-j,2n-i+j}+ ... .
\end{equation}
$a=q^{-\frac{1}{2}(\lambda^{i},\lambda^{2n-i})}$ is from the rotation of $z_{1}$ and $z_{2}$. $$b=a\cdot q^{\frac{1}{4}[(\lambda^{i-j},\lambda^{i-j}-\lambda^{i})+(\lambda^{2n-i},\lambda^{i-j}-\lambda^{i})]+\frac{1}{2}(j-1)}$$ is from moving of $\{\alpha_{i-j+1}, ..., \alpha_{i}\}$ from $z_{2}$ to $z_{1}$.
$c=b\cdot q^{-(\lambda^{2n-i},\lambda^{i-j}-\lambda^{i})}$ is from the anti-clockwise rotation around $z_{1}$ and $2n-i$ primary roots near $z_{1}$.

Thus,
\begin{equation}\label{eqbr1}
\boldsymbol{B}_{2n-i,i}^{i-j,2n-i+j}=(-1)^{j}(b-c)
=(-1)^{j}q^{\frac{j}{2}}(q^{\frac{1}{2}}-q^{-\frac{1}{2}}).
\end{equation}

\begin{figure}
\centering\includegraphics[width=4.5cm]{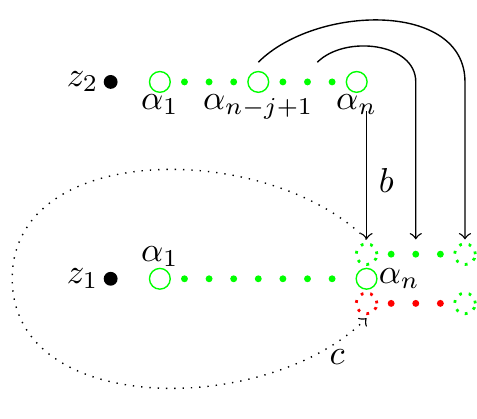}\\
\caption{$0<j\leq i=n$.}
 \label{fig:cw-b1}
 \end{figure}
ii)$0<j\leq i=n$

\begin{equation}
\boldsymbol{B}J_{n,n}
=aJ_{n,n}+((-1)^{j}b+(-1)^{j+1}c)J_{n-j,n+j}+ ... .
\end{equation}
$a=q^{-\frac{1}{2}(\lambda^{n},\lambda^{2n-n})}$ is from the rotation of $z_{1}$ and $z_{2}$.
$$b=a\cdot q^{\frac{1}{4}[(\lambda^{n-j},\lambda^{n-j}-\lambda^{n})+(\lambda^{2n-n},\lambda^{n-j}-\lambda^{n})]+\frac{1}{4}(\alpha_{n},\alpha_{n})+\frac{1}{2}(j-1)}$$ is from moving of $\{\alpha_{n-j+1}, ..., \alpha_{n}\}$ from $z_{2}$ to the position indicated by the green dotted circles as shown in figure \ref{fig:cw-b1}. Additional anti-clockwise rotation to the position indicated by the red dotted circles gives $$c=b\cdot q^{-(\lambda^{2n-n},\lambda^{n-j}-\lambda^{n})-\frac{1}{2}(\alpha_{n},\alpha_{n})}.$$
This variation is different from type $III$ by the position of two $\alpha_{n}$. The variation of $\alpha_{n}$ is of type $I$. In this sense, we call the variation a combination of type $I$ and $III$.
Thus,
\begin{equation}
\boldsymbol{B}_{n,n}^{n-j,n+j}=(-1)^{j}(b-c)
=(-1)^{j}q^{\frac{j}{2}}(1-q^{-\frac{1}{2}}).
\end{equation}

\begin{figure}
\centering\includegraphics[width=6cm]{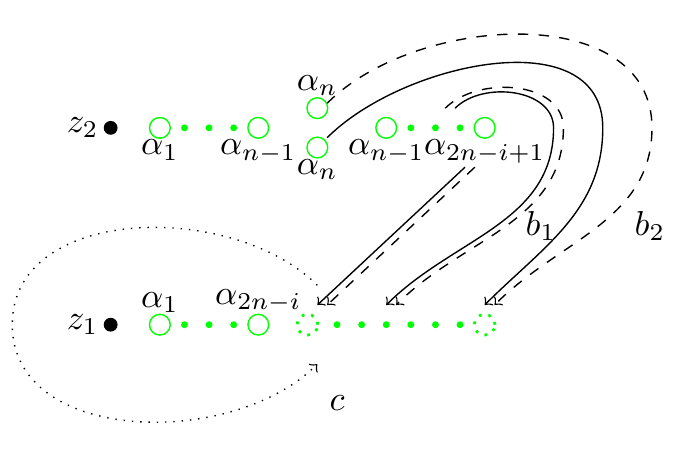}\\
\caption{ $i>n\& i-j=n$.}
 \label{fig:cw-b2}
 \end{figure}
iii) $i>n\& i-j=n$
As shown in figure \ref{fig:cw-b2}, because of two $\alpha_{n}$, there are two possible cases of moving. They give coefficients $b_{1}$ and $b_{2}$.
\begin{equation}
  \boldsymbol{B}J_{2n-i,i}
=aJ_{i,2n-i}+((-1)^{j}(b_{1}+b_{2})+(-1)^{j+1}(c_{1}+c_{2}))J_{n,n}+ ....
\end{equation}
$b_{1}+b_{2}=a\cdot q^{\frac{1}{4}[(\lambda^{i-j},\lambda^{i-j}-\lambda^{i})+(\lambda^{2n-i},\lambda^{i-j}-\lambda^{i})]+\frac{1}{2}(j-1)}(q^{\frac{1}{4}(\alpha_{n},\alpha_{n})}+q^{-\frac{1}{4}(\alpha_{n},\alpha_{n})})$.
Additional anti-clockwise rotation around $z_{1}$ and other $2n-i$  primary roots gives
$c_{1}+c_{2}=(b_{1}+b_{2})\cdot q^{-(\lambda^{2n-i},\lambda^{i-j}-\lambda^{i})}$. The variation is also a combination of type $I$ and $III$. Thus,
\begin{equation}
\boldsymbol{B}_{2n-i,i}^{n,n}
=(-1)^{j}q^{\frac{j}{2}+\frac{1}{4}}(q^{\frac{1}{4}}+q^{-\frac{1}{4}})(1-q^{-1}).
\end{equation}

iv) $i>n\& i-j>n$

The variation is of type $III$. $\boldsymbol{B}_{2n-i,i}^{i-j,2n-i+j}$ is exactly the same as \eqref{eqbr1}.
$$\boldsymbol{B}_{2n-i,i}^{i-j,2n-i+j}
=(-1)^{j}q^{\frac{j}{2}}(q^{\frac{1}{2}}-q^{-\frac{1}{2}}).
$$

v) $i>n\& i-j<n\&i-j\neq 2n-i$

The variation is of type $III$,
$b=a\cdot q^{\frac{1}{4}[(\lambda^{i-j},\lambda^{i-j}-\lambda^{i})+(\lambda^{2n-i},\lambda^{i-j}-\lambda^{i})]+\frac{1}{2}(j-2)}$ and
$c=b\cdot q^{-(\lambda^{2n-i},\lambda^{i-j}-\lambda^{i})}$.
\begin{equation}
\boldsymbol{B}_{2n-i,i}^{i-j,2n-i+j}
=(-1)^{j}q^{\frac{j}{2}}(1-q^{-1}).
\end{equation}

vi) $i>n\& i-j=2n-i\& i\neq n+1$
The variation is of type $IV$. $i-j=2n-i$ means that the moving primary roots are of double multiplicity and symmetrically distributed. Same as type $III$, moving of them from $z_{2}$ to $z_{1}$ is accompanied by the self clockwise rotation of $\pi$.
But because of double multiplicity and symmetric distribution, the order of these primary roots will give three additional terms in the wall-crossing formula, as is shown in figure \ref{fig:t4}.
\begin{equation}
 \begin{split}
\boldsymbol{B}J_{2n-i,i}
&=aJ_{i,2n-i}+((-1)^{j}b+(-1)^{j+1}c_{1}-c_{2}+c_{3})J_{2n-i,i}+ ...;\\
b&=a\cdot q^{\frac{1}{4}[(\lambda^{i-j},\lambda^{i-j}-\lambda^{i})+(\lambda^{2n-i},\lambda^{i-j}-\lambda^{i})]+n-i+j-\frac{3}{2}};\\ c_{1}&=b\cdot q^{-(\lambda^{2n-i},\lambda^{2n-i}-\lambda^{2n-i+1})};\\
c_{2}&=c_{3}\cdot q^{-(\lambda^{2n-i},\lambda^{2n-i}-\lambda^{2n-i+1})};\\
c_{3}&=a\cdot q^{\frac{1}{4}[(\lambda^{i-j},\lambda^{i-j}-\lambda^{i})+(\lambda^{2n-i},\lambda^{i-j}-\lambda^{i})]-(\lambda^{2n-i},\lambda^{i-j}-\lambda^{i})};\\
\boldsymbol{B}_{2n-i,i}^{2n-i,i}
&=((-1)^{j}q^{n-i+j-\frac{1}{2}}-1)(q^{\frac{1}{2}}-q^{-\frac{1}{2}}).\\
\end{split}
\end{equation}

vii) $i=n+1\& j=2$

The variation is of type $I$.
\begin{equation}
 \begin{split}
\boldsymbol{B}J_{n-1,n+1}=&aJ_{n+1,n-1}+((-1)^{2}b+(-1)^{3}(c_{1}+c_{2})+(-1)^{4}c_{3})J_{n-1,n+1}+ ...;\\
b=&a\cdot q^{\frac{1}{4}[(\lambda^{i-j},\lambda^{i-j}-\lambda^{i})+(\lambda^{2n-i},\lambda^{i-j}-\lambda^{i})]-\frac{1}{2}};\\
c_{1}+c_{2}=&b\cdot q^{-(\lambda^{2n-i},\lambda^{2n-i}-\lambda^{2n-i+1})}(1+q^{\frac{1}{2}(\alpha_{n},\alpha_{n})});\\ c_{3}=&b\cdot q^{-2(\lambda^{2n-i},\lambda^{2n-i}-\lambda^{2n-i+1})+\frac{1}{2}(\alpha_{n},\alpha_{n})};\\ \boldsymbol{B}_{n-1,n+1}^{n-1,n+1}
=&q-1-q^{\frac{1}{2}}+q^{-\frac{1}{2}}.\\
\end{split}
\end{equation}

In summarization,

\begin{equation}\label{Bn1}
\begin{split}
  &\boldsymbol{B}_{2n-i,i}^{i-j,2n-i+j}\\
=&\left\{
                            \begin{array}{ll}
                             (-1)^{j}q^{\frac{j}{2}}(q^{\frac{1}{2}}-q^{-\frac{1}{2}}) , & \hbox{$i>n\& i-j>n$;} \\
                                                                                         & \hbox{$0<j\leq i\leq n-1$;} \\
                              (-1)^{j}q^{\frac{j}{2}}(1-q^{-\frac{1}{2}}), & \hbox{$0<j\leq i= n$;} \\
                              (-1)^{j}q^{\frac{j}{2}+\frac{1}{4}}(q^{\frac{1}{4}}+q^{-\frac{1}{4}})(1-q^{-1}), & \hbox{$i>n\& i-j=n$;} \\
                              (-1)^{j}q^{\frac{j-1}{2}}(q^{\frac{1}{2}}-q^{-\frac{1}{2}}), & \hbox{$i>n\& i-j<n\&i-j\neq2n-i$;} \\
                              (q^{\frac{1}{2}}-q^{-\frac{1}{2}})((-1)^{j}q^{n-i+j-\frac{1}{2}}-1), & \hbox{$i>n\& i-j=2n-i$.}
                            \end{array}
                          \right.
\end{split}
\end{equation}

Monodromy representation is as following:

\begin{equation}\label{Bn2}
\begin{split}
  &\boldsymbol{B}_{2n-a,a}^{2n-b,b}\\
=&\left\{
                        \begin{array}{ll}
                          (-1)^{a+n}q^{\frac{a-n}{2}}(q^{\frac{1}{2}}-q^{-\frac{1}{2}})(1+q^{-\frac{1}{2}}), & \hbox{$a>n\&b=n$;} \\
                          (-1)^{a+b}(q^{\frac{1}{2}}-q^{-\frac{1}{2}})q^{-n+\frac{a+b}{2}}, & \hbox{$a<n\&b>n \hbox{  or  } a>n\&b<n$;} \\
                          (q^{\frac{1}{2}}-q^{-\frac{1}{2}})((-1)^{a+b}q^{-n+\frac{a+b-1}{2}}-\delta_{a,b}), & \hbox{$a>n\&b>n$;} \\
                         (-1)^{n+b}q^{\frac{b-n}{2}}(1-q^{-\frac{1}{2}}) , & \hbox{$a=n\&b>n$.}
                        \end{array}
                      \right.
\end{split}
\end{equation}

We have given the description for four types of variations in detail. In the following cases of $C_{n}$ and $D_{n}$, the methods for proofs are similar, so we omit them except for some extraordinary cases.

\subsubsection{$C_{n}$}\label{subsubsct:C}
Denote the weights of the fundamental representation of $C_{n}$ Lie algebra by $$\lambda^{i}=\left\{
                                                                                        \begin{array}{ll}
                                                                                          \lambda-\sum_{j=1}^{i}\alpha_{j}, & \hbox{$i\leq n$;} \\
                                                                                          \lambda-\sum_{j=1}^{n}\alpha_{j}-\sum_{j=2n-i}^{n}\alpha_{j}, & \hbox{$i> n$.}
                                                                                        \end{array}
                                                                                      \right.$$
The inner products of them are:
\begin{equation}
(\lambda^{s},\lambda^{t})=\left\{
                            \begin{array}{ll}
                              \frac{1}{2}, & \hbox{$s=t$;} \\
                              0, & \hbox{$s+t\neq2n-1,s\neq t$;} \\
                              -\frac{1}{2}, & \hbox{$s+t=2n-1$,}
                            \end{array}
                          \right.
\end{equation}where $s,t=0,1,..,2n-1$.

For any weight vector $v_{\lambda^{l}}\in V_{\omega_{1}}$, Yang-Yang function ${\bm W}_{c}({\bm w},z,\omega_{1},l)$ corresponding to it is as following:
\begin{equation}
\begin{split}
{\bm W}_{c}({\bm w},z,\omega_{1},l)
=\sum _{j=1}^{l}(\alpha_{i_{j}},\omega_{1})\log \left( w_{j}-z\right) -&\sum_{1\leq j< k\leq l}(\alpha_{i_{j}},\alpha_{i_{k}}) \log \left( w_{j}-w_{k}\right)\\
-&c(\sum_{j=1}^{l}(\alpha_{i_{j}},\rho)w_{j}-(\omega_{1},\rho)z),
\end{split}
\end{equation}where $$i_{j}=\left\{
                                      \begin{array}{ll}
                                        j, & \hbox{$j\leq n$;} \\
                                        2n-j, & \hbox{$n<j\leq 2n-1$.}
                                      \end{array}
                                    \right.
$$
By this notation, when $l\geq n+1$, $\{w_{k}, w_{2n-k}\}_{2n-l\leq k\leq n-1}$ are pairs of symmetric coordinates of $\alpha_{k}$ in the function. Define $\bar{w}_{k}=w_{k}+w_{2n-k}$ and $\Delta_{k}=(w_{k}-w_{2n-k})^{2}$, $2n-l\leq k\leq n-1$.

\begin{lemma} For the fundamental representation $V_{\omega_{1}}$ of $C_{n}$ Lie algebra,  the solutions of the critical point equation \eqref{CE2} of the corresponding Yang-Yang functions ${\bm W}_{c}({\bm w},0,\omega_{1},l)$ are as following:

When $l<n$, $$w_{j}=\sum_{i=1}^{j}\frac{1}{c(l-i+1)}\quad j=1,...,l.$$

When $l=n$, $$w_{j}=\sum_{i=1}^{j}\frac{1}{c(n-i+2)}\quad j=1,...,n-1,$$
$$w_{n}=\frac{1}{c}+\sum_{i=1}^{n-1}\frac{1}{c(n-i+2)}.$$

When $l\geq n+1$, $$w_{k}=\sum_{i=1}^{j}\frac{1}{c(l+2-i)}$$ for $k=1,...,2n-l-1,$
$$w_{n}=\sum_{i=1}^{2n-l-1}\frac{1}{c(l+2-i)}+\sum_{i=2n-l}^{n}\frac{1}{c(l+1-i)}.$$
For $2n-l\leq k\leq n-1$,
\begin{equation}
\begin{split}
\bar{w}_{k}=&\sum_{i=1}^{2n-l-1}\frac{2}{c(l-i+2)}+\sum_{i=2n-l}^{k-l}\frac{2}{c(l-i+1)}+\sum_{i=k}^{n}\frac{1}{c(l-i+1)}\\&+\sum_{i=k}^{n}\frac{1}{c(l-2n+i+1)},
\end{split}\end{equation}

\begin{equation}
\begin{split}
\Delta_{k}=&[\sum_{i=k}^{n}\frac{1}{c(l-i+1)}+\sum_{i=k}^{n}\frac{1}{c(l-2n+i+1)}]\\
&\times[\sum_{i=k}^{n}\frac{1}{c(l-i+1)}+\sum_{i=k}^{n}\frac{1}{c(l-2n+i+1)}-\frac{4}{c(2l-2n+1)}].
\end{split}\end{equation}
\end{lemma}

\begin{lemma} For any $2n-l\leq k\leq n-1$, $\Delta_{k}>0$. Assume $w_{k}<w_{2n-k}$ for $k=2n-l,...,n-1$, then for any $l>0$, the coordinates of critical solutions satisfy the following order:
$$0<w_{1}<w_{2}<...<w_{l}.$$

\end{lemma}

By this lemma, $z$ and the critical coordinates $\{w_{j}\}_{j=1,...,l}$ of ${\bm W}_{c}({\bm w},z,\omega_{1},l)$ are on the same horizontal line of ${\bm W}$ plane as shown in figure \ref{fig:dc}.

\begin{figure}
 \centering\includegraphics[width=3cm]{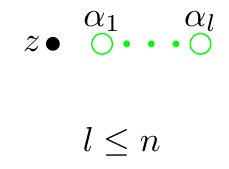}
\centering\includegraphics[width=6cm]{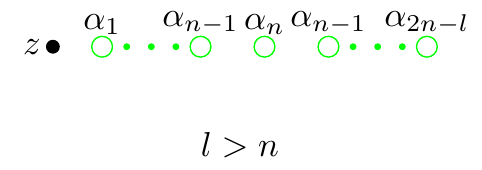}\\
\caption{Coordinates distribution of the $C_{n}$ critical point on ${\bm W}_{c}({\bm w},{\bm z},{\bm \lambda},{\bm l})$ plane near $z$ .}
 \label{fig:dc}
 \end{figure}

To derive monodromy representation, firstly we consider the case $v_{\lambda^{i}}\otimes v_{\lambda^{j}}\in V_{\omega_{1}}\otimes V_{\omega_{1}}, \quad i+j\neq 2n-1$.

i)
When $i+j\neq 2n-1\& i\geq j $, there is no wall-crossing, $$\boldsymbol{B}J_{i,j}=q^{-\frac{1}{2}(\lambda^{i},\lambda^{j})}J_{j,i}.$$

ii)
When $i+j\neq 2n-1\& i<j$,  the variation of critical point in wall-crossing is of type $II$.

\begin{equation}
\begin{split}
  &\boldsymbol{B}J_{i,j}\\
=&q^{-\frac{1}{2}(\lambda^{i},\lambda^{j})}(J_{j,i}+q^{\frac{1}{2}(\lambda^{i},\lambda^{i}-\lambda^{j})}(q^{-(\lambda^{i},\lambda^{i}-\lambda^{j})}-1)J_{i,j})\\
=&J_{j,i}+(q^{-\frac{1}{2}}-q^{\frac{1}{2}})J_{i,j}.\\
\end{split}
\end{equation}

In sum,

\begin{equation}
  \boldsymbol{B}J_{i,j}=\left\{
             \begin{array}{ll}
               q^{-\frac{1}{4}}J_{j,i}, & \hbox{$i+j\neq 2n-1\&i=j$;} \\
               J_{j,i}, & \hbox{$i+j\neq 2n-1\&i> j$;} \\
               J_{j,i}+(q^{-\frac{1}{4}}-q^{\frac{1}{4}})J_{i,j}, & \hbox{$i+j\neq 2n-1\&i<j$.}
             \end{array}
           \right.
\end{equation}

When $i+j=2n-1$, the critical point equation of two singularities $z_{1}$ and $z_{2}$ with $c=0$ is as following:
\begin{equation}\label{equ:CNSB}
  \left\{
    \begin{aligned}
      0&=\frac{-\frac{1}{2}}{w^{1}_{1}-z_{1}}+\frac{-\frac{1}{2}}{w^{1}_{1}-z_{2}}+\frac{1}{w^{1}_{1}-w^{2}_{1}}+\frac{-\frac{1}{2}}{w^{1}_{1}-w^{1}_{2}}+\frac{-\frac{1}{2}}{w^{1}_{1}-w^{2}_{2}}\\
      0&=\frac{-\frac{1}{2}}{w^{2}_{1}-z_{1}}+\frac{-\frac{1}{2}}{w^{2}_{1}-z_{2}}+\frac{1}{w^{2}_{1}-w^{1}_{1}}+\frac{-\frac{1}{2}}{w^{2}_{1}-w^{1}_{2}}+\frac{-\frac{1}{2}}{w^{2}_{1}-w^{2}_{2}}\\
      ...   \\
      0&=\frac{1}{w^{1}_{k}-w^{2}_{k}}+\frac{-\frac{1}{2}}{w^{1}_{k}-w^{1}_{k-1}}+\frac{-\frac{1}{2}}{w^{1}_{k}-w^{2}_{k-1}}+\frac{-\frac{1}{2}}{w^{1}_{k}-w^{1}_{k+1}}+\frac{-\frac{1}{2}}{w^{1}_{k}-w^{2}_{k+1}}\\
      0&=\frac{1}{w^{2}_{k}-w^{1}_{k}}+\frac{-\frac{1}{2}}{w^{2}_{k}-w^{1}_{k-1}}+\frac{-\frac{1}{2}}{w^{2}_{k}-w^{2}_{k-1}}+\frac{-\frac{1}{2}}{w^{2}_{k}-w^{1}_{k+1}}+\frac{-\frac{1}{2}}{w^{2}_{k}-w^{2}_{k+1}}\\
       ...  \\
      0&=\frac{1}{w^{1}_{n-1}-w^{2}_{n-1}}+\frac{-\frac{1}{2}}{w^{1}_{n-1}-w^{1}_{n-2}}+\frac{-\frac{1}{2}}{w^{1}_{n-1}-w^{2}_{n-2}}+\frac{-1}{w^{1}_{n-1}-w_{n}}\\
      0&=\frac{1}{w^{2}_{n-1}-w^{1}_{n-1}}+\frac{-\frac{1}{2}}{w^{2}_{n-1}-w^{1}_{n-2}}+\frac{-\frac{1}{2}}{w^{2}_{n-1}-w^{2}_{n-2}}+\frac{-1}{w^{2}_{n-1}-w_{n}}\\
      0&=\frac{-1}{w_{n}-w^{1}_{n-1}}+\frac{-1}{w_{n}-w^{2}_{n-1}},\\
    \end{aligned}
  \right.
\end{equation}
where $2\leq k\leq n-2$.
By lemma \ref{rlemma}, its solution is as following: $$\bar{w}_{1}=\bar{w}_{k}=\bar{w}_{n-1}=z_{1}+z_{2},\quad \bar{w}_{n}=\frac{z_{1}+z_{2}}{2}$$ and $$w^{1}_{l}w^{2}_{l}=z_{1}z_{2}+\frac{(z_{1}-z_{2})^{2}l(2n+1-l)}{4n(n+1)},\quad 1\leq l\leq n-1.$$ The distribution of the coordinates of critical point is shown in figure \ref{fig:nsbc}. Similar to the case of one singularity in figure \ref{fig:dc}, the coordinates are symmetric on the same line connecting $z_{1}$ and $z_{2}$, which gives coordinates of the thimble and singularities $z_{1}$ and $z_{2}$ an order:$$z_{1}<w^{1}_{1}<...<w^{1}_{n-1}<w_{n}<w^{2}_{n-1}<...<w^{2}_{1}<z_{2}.$$
Because of the existence of the thimble above, there is a wall-crossing whenever $i\neq 0$ and its homotopic class will keep this order.

For $v_{\lambda^{i}}\otimes v_{\lambda^{j}}\in V_{\omega_{1}}\otimes V_{\omega_{1}}, \quad i+j=2n-1$, the coefficients of equation \eqref{equ:cof} are $$\boldsymbol{B}_{2n-1-i,i}^{i,2n-1-i}=q^{-\frac{1}{2}(\lambda^{2n-1-i},\lambda^{i})}=q^{\frac{1}{4}}$$ and
 $\boldsymbol{B}_{2n-1-i,i}^{i-j,2n-1-i+j} (j>0)$ can be computed as in the following cases.

i) $0<j\leq i<n$

The variation is of type $III$.

\begin{equation}\label{eqbr1C}
\begin{split}
b&=a\cdot q^{\frac{1}{4}[(\lambda^{i-j},\lambda^{i-j}-\lambda^{i})+(\lambda^{2n-1-i},\lambda^{i-j}-\lambda^{i})]+\frac{1}{4}(j-1)};\\
c&=b\cdot q^{-(\lambda^{2n-1-i},\lambda^{i-j}-\lambda^{i})};\\
\boldsymbol{B}_{2n-1-i,i}^{i-j,2n-1-i+j}
&=(-1)^{j}q^{\frac{j+1}{4}}(1-q^{-\frac{1}{2}}).\\
\end{split}
\end{equation}

ii)$1<j\leq i=n$

The variation is of type $III$.
\begin{equation}\label{eqbr2C}
\begin{split}
b&=a\cdot q^{\frac{1}{4}[(\lambda^{i-j},\lambda^{i-j}-\lambda^{i})+(\lambda^{2n-1-i},\lambda^{i-j}-\lambda^{i})]+\frac{1}{4}j};\\
c&=b\cdot q^{-(\lambda^{2n-1-i},\lambda^{i-j}-\lambda^{i})};\\
\boldsymbol{B}_{2n-1-i,i}^{i-j,2n-1-i+j}
&=(-1)^{j}q^{\frac{j+2}{4}}(1-q^{-\frac{1}{2}}).\\
\end{split}
\end{equation}

iii) $i=n\& j=1$

The variation is of type $III$.
\begin{equation}
\begin{split}
b&=a\cdot q^{\frac{1}{4}[(\lambda^{i-j},\lambda^{i-j}-\lambda^{i})+(\lambda^{2n-1-i},\lambda^{i-j}-\lambda^{i})]};\\
c&=b\cdot q^{-(\lambda^{n-1},\lambda^{n-1}-\lambda^{n})};\\
\boldsymbol{B}_{n-1,n}^{n-1,n}
&=q^{\frac{3}{4}}(-1+q^{-1}).\\
\end{split}
\end{equation}

iv) $n<i\leq 2n-1\& i-j\geq n$

Same as \eqref{eqbr1C}.

v)  $n<i\leq 2n-1\& i-j=n-1$

Same as \eqref{eqbr2C}.

vi) $i>n\& i-j<n-1\& i-j\neq 2n-1-i$

The variation is of type $III$.
\begin{equation}
\begin{split}
b&=a\cdot q^{\frac{1}{4}[(\lambda^{i-j},\lambda^{i-j}-\lambda^{i})+(\lambda^{2n-1-i},\lambda^{i-j}-\lambda^{i})]+\frac{1}{4}j};\\
c&=b\cdot q^{-(\lambda^{2n-1-i},\lambda^{i-j}-\lambda^{i})};\\
\boldsymbol{B}_{2n-1-i,i}^{i-j,2n-1-i+j}
&=(-1)^{j}q^{\frac{j+2}{4}}(1-q^{-\frac{1}{2}}).\\
\end{split}
\end{equation}

vii) $i>n \& i-j<n-1\&i-j=2n-1-i$

The variation is of type $IV$.
\begin{equation}
\begin{split}
b=&a\cdot q^{\frac{1}{4}[(\lambda^{i-j},\lambda^{i-j}-\lambda^{i})+(\lambda^{2n-1-i},\lambda^{i-j}-\lambda^{i})]+\frac{1}{2}(n-i+j-1)};\\
c_{1}=&b\cdot q^{-(\lambda^{2n-1-i},\lambda^{i-j}-\lambda^{i-j+1})};\\
c_{2}=&a\cdot q^{\frac{1}{4}[(\lambda^{i-j},\lambda^{i-j}-\lambda^{i})+(\lambda^{2n-1-i},\lambda^{i-j}-\lambda^{i})]}\\ &\cdot q^{-(\lambda^{2n-1-i},\lambda^{i-j}-\lambda^{i})+(\lambda^{2n-1-i},\lambda^{i-j}-\lambda^{i-j+1})};\\
c_{3}=&a\cdot q^{\frac{1}{4}[(\lambda^{i-j},\lambda^{i-j}-\lambda^{i})+(\lambda^{2n-1-i},\lambda^{i-j}-\lambda^{i})]-(\lambda^{2n-1-i},\lambda^{i-j}-\lambda^{i})};\\
\boldsymbol{B}_{2n-1-i,i}^{i-j,2n-1-i+j}
=&(-1)^{j}q^{\frac{1}{2}(n-i+j)+\frac{1}{4}}(1-q^{-\frac{1}{2}})+q^{-\frac{1}{4}}(1-q^{\frac{1}{2}}).\\
\end{split}
\end{equation}

In sum,
\begin{equation}\label{Cn1}
\begin{split}
&\boldsymbol{B}_{2n-1-i,i}^{i-j,2n-1-i+j}\\
=&\left\{
                              \begin{array}{ll}
                                 q^{\frac{j+1}{4}}((-1)^{j}+(-1)^{j+1}q^{-\frac{1}{2}}), & \hbox{$0<j\leq i<n$;} \\
                                q^{\frac{3}{4}}(q^{-1}-1), & \hbox{$1=j\leq i=n$;} \\
                                q^{\frac{j+2}{4}}((-1)^{j}+(-1)^{j+1}q^{-\frac{1}{2}}), & \hbox{$1<j\leq i=n$;} \\
                                q^{\frac{j+1}{4}}((-1)^{j}+(-1)^{j+1}q^{-\frac{1}{2}}), & \hbox{$n<i\leq 2n-1\&j>0\& i-j\geq n$;} \\
                                q^{\frac{j+2}{4}}((-1)^{j}+(-1)^{j+1}q^{-\frac{1}{2}}), & \hbox{$n<i\leq 2n-1\& i-j=n-1$;} \\
                                q^{\frac{j+2}{4}}((-1)^{j}+(-1)^{j+1}q^{-\frac{1}{2}}), & \hbox{$n<i\leq 2n-1$} \\
                                                                                        & \hbox{$\& 2n-1-i<i-j<n-1$;} \\
                                (1-(-1)^{j}q^{\frac{n-i+j}{2}})(q^{-\frac{1}{4}}-q^{\frac{1}{4}}), & \hbox{$n<i\leq 2n-1$} \\
                                                                    & \hbox{$\&2n-1-i=i-j<n-1$;} \\
                                q^{\frac{j+2}{4}}((-1)^{j}+(-1)^{j+1}q^{-\frac{1}{2}}), & \hbox{$n<i\leq 2n-1\& i-j<2n-1-i$.}
                              \end{array}
                            \right.
\end{split}
\end{equation}

Monodromy representation is as following:

\begin{equation}\label{Cn2}
\begin{split}
  &\boldsymbol{B}_{2n-1-a,a}^{2n-1-b,b}\\
=&\left\{
                            \begin{array}{ll}
                              (-1)^{a+b+1}q^{-\frac{2n-a-b-1}{4}}(q^{\frac{1}{4}}-q^{-\frac{1}{4}}), & \hbox{$a\leq n-1\& b\geq n$;} \\
&\hbox{$a\geq n\& b\leq n-1$;} \\
                              (-1)^{a+b+1}q^{-\frac{2n-a-b-2}{4}}(q^{\frac{1}{4}}-q^{-\frac{1}{4}})\\
-\delta_{a,b}(q^{\frac{1}{4}}-q^{-\frac{1}{4}}), & \hbox{$a\geq n\& b\geq n$.}
                            \end{array}
                          \right.
\end{split}
\end{equation}

\subsubsection{$D_{n}$}\label{subsubsct:D}

Denote the weights of the fundamental representation of $D_{n}$ Lie algebra by

\begin{align}
&\lambda^{i}=\left\{
                                                                                        \begin{array}{ll}
                                                                                          \lambda-\sum_{j=1}^{i}\alpha_{j}, & \hbox{$i\leq n$;} \\
                                                                                          \lambda-\sum_{j=1}^{n}\alpha_{j}-\sum_{j=2n-i}^{n}\alpha_{j}, & \hbox{$i> n$.}\\
                                                                                        \end{array}
                                                                                      \right.\\
&\lambda^{{n-1}^{'}}=\lambda-\sum_{i=1}^{n-2}\alpha_{i}-\alpha_{n}.
\end{align}

For convenience of discussion, we define the order of $i=0,1,..,n-2,n-1,n-1',n,..2n-2$ as following
\begin{equation}
o(i)=\left\{
       \begin{array}{ll}
         i, & \hbox{$i=0,1,..,n-2,n-1$;} \\
         n, & \hbox{$i=n-1'$;} \\
         i+1, & \hbox{$i=n,n+1,..,2n-2$.}
       \end{array}
     \right.
\end{equation}

The inner products of them are:
\begin{equation}
(\lambda^{s},\lambda^{t})=\left\{
                            \begin{array}{ll}
                             1, & \hbox{$o(s)+o(t)\neq 2n-1,s=t$;} \\
                              0, & \hbox{$o(s)+o(t)\neq 2n-1,s\neq t$;} \\
                              -1, & \hbox{$o(s)+o(t)=2n-1$,}
                            \end{array}
                          \right.
\end{equation}where $s,t=0,1,..,n-2,n-1,n-1^{'},n,n+1,..,2n-2$.

Every weight vector $v_{\lambda^{l}}\in V_{\lambda}$ has Yang-Yang function ${\bm W}_{c}({\bm w},z,\omega_{1},l)$ corresponding to it.
When $l\neq n-1^{'}$,

\begin{equation}
\begin{split}
{\bm W}_{c}({\bm w},z,\omega_{1},l)
=\sum _{j=1}^{l}(\alpha_{i_{j}},\omega_{1})\log \left( w_{j}-z\right) -&\sum_{1\leq j< k\leq l}(\alpha_{i_{j}},\alpha_{i_{k}}) \log \left( w_{j}-w_{k}\right)\\
-&c(\sum_{j=1}^{l}(\alpha_{i_{j}},\rho)w_{j}-(\omega_{1},\rho) z),
\end{split}
\end{equation}where $$i_{j}=\left\{
                                      \begin{array}{ll}
                                        j, & \hbox{$j\leq n$;} \\
                                        2n-1-j, & \hbox{$n+1\leq j\leq 2n-2$.}
                                      \end{array}
                                    \right.
$$
When $l= n-1^{'}$,
\begin{equation}
\begin{split}
{\bm W}_{c}({\bm w},z,\omega_{1},l)
=\sum _{j=1}^{n-1}(\alpha_{i_{j}},\omega_{1})\log \left( w_{j}-z\right) -&\sum_{ j< k\leq n-1}(\alpha_{i_{j}},\alpha_{i_{k}}) \log \left( w_{j}-w_{k}\right)\\
-&c(\sum_{j=1}^{n-1}(\alpha_{i_{j}},\rho)w_{j}-(\omega_{1},\rho) z),
\end{split}
\end{equation}where $$i_{j}=\left\{
                                      \begin{array}{ll}
                                        j, & \hbox{$j=1,...,n-2$;} \\
                                        n, & \hbox{$j=n-1$.}
                                      \end{array}
                                    \right.
$$

By this notation, when $l\geq n+1$, $\{w_{k}, w_{2n-1-k}\}_{2n-1-l\leq k\leq n-2}$ are pairs of symmetric coordinates of $\alpha_{k}$ in the function. Define $\bar{w}_{k}=w_{k}+w_{2n-1-k}$ and $\Delta_{k}=(w_{k}-w_{2n-1-k})^{2}$, $2n-1-l\leq k\leq n-2$.

\begin{lemma} For the fundamental representation $V_{\omega_{1}}$ of $D_{n}$ Lie algebra,   the solutions of the critical point equation \eqref{CE2} of the corresponding Yang-Yang functions ${\bm W}_{c}({\bm w},0,\omega_{1},l)$ are as following:

When $l\leq n-1$, $$
w_{j}=\sum_{i=1}^{j}\frac{1}{c(l-i+1)}\quad j=1,...,l;$$

When $l=n-1^{'}$, $$w_{j}=\sum_{i=1}^{j}\frac{1}{c(n-i)}\quad j=1,...,n-1;$$

When $l=n$, $$w_{j}=\sum_{i=1}^{j}\frac{1}{c(n+1-i)}\quad j=1,...,n-2,$$ and $$w_{n-1}=w_{n}=\frac{1}{c}+\sum_{i=1}^{n-2}\frac{1}{c(n+1-i)};$$

When $l\geq n+1$, $$w_{k}=\sum_{i=1}^{k}\frac{1}{c(l+1-i)}$$ for $k=1,...,2n-2-l,$
$$w_{n-1}=w_{n}=\sum_{i=1}^{2n-2-l}\frac{1}{c(l+1-i)}+\sum_{i=2n-1-l}^{n-2}\frac{1}{c(l-i)}+\frac{1}{c(l+1-n)},$$

\begin{equation}\begin{split}
\bar{w}_{k}=&\sum_{i=1}^{2n-l-2}\frac{2}{c(l+1-i)}+\sum_{i=2n-1-l}^{k-1}\frac{1}{c(l-i)}+\sum_{i=2n-1-l}^{n-2}\frac{1}{c(l-i)}\\
&+\frac{2}{c(l-n+1)}+\sum_{i=k}^{n-2}\frac{1}{c(l+i+2-2n)},\end{split}\end{equation} and
\begin{equation}\begin{split}\Delta_{k}=&[\sum_{i=k}^{n-2}\frac{1}{c(l-i)}+\frac{2}{c(l-n+1)}+\sum_{i=k}^{n-2}\frac{1}{c(l+i+2-2n)}]^2\\&+\sum_{i=k}^{n-2}\frac{1}{c(l-i)}
+\frac{2}{c(l-n+1)}+\sum_{i=k}^{n-2}\frac{1}{c(l+i+2-2n)}-\frac{4}{c(2l-2n+1)},\end{split}\end{equation}
for $2n-1-l\leq k\leq n-2$.
\end{lemma}
\begin{lemma} $\Delta_{k}>0$ for $k=2n-1-l,...,n-2$. Assume $w_{k}<w_{2n-1-k}$, then the coordinates of critical solutions satisfy the following order:

When $l\leq n-1$, $$0<w_{1}<w_{2}<...<w_{l};$$

When $l=n-1^{'}$, $$0<w_{1}<w_{2}<...<w_{n-1};$$

When $l\geq n$,
$$0<w_{1}<...<w_{n-2}<w_{n-1}=w_{n}<w_{n+1}<...<w_{l}.$$

\end{lemma}

By this lemma, $z$ and the critical coordinates $\{w_{j}\}_{j=1,...,l}$ of
${\bm W}_{c}({\bm w},z,\omega_{1},l)$ are on the same horizontal line of ${\bm W}$ plane as shown in figure \ref{fig:dd}.

When $l=2n-2$, the critical point equation of two singularities $z_{1}$ and $z_{2}$ with $c=0$ is as following:
\begin{equation}\label{equ:DNSB}
  \left\{
    \begin{aligned}
      0&=\frac{-1}{w^{1}_{1}-z_{1}}+\frac{-1}{w^{1}_{1}-z_{2}}+\frac{2}{w^{1}_{1}-w^{2}_{1}}+\frac{-1}{w^{1}_{1}-w^{1}_{2}}+\frac{-1}{w^{1}_{1}-w^{2}_{2}}\\
      0&=\frac{-1}{w^{2}_{1}-z_{1}}+\frac{-1}{w^{2}_{1}-z_{2}}+\frac{2}{w^{2}_{1}-w^{1}_{1}}+\frac{-1}{w^{2}_{1}-w^{1}_{2}}+\frac{-1}{w^{2}_{1}-w^{2}_{2}}\\
      ...   \\
      0&=\frac{2}{w^{1}_{k}-w^{2}_{k}}+\frac{-1}{w^{1}_{k}-w^{1}_{k-1}}+\frac{-1}{w^{1}_{k}-w^{2}_{k-1}}+\frac{-1}{w^{1}_{k}-w^{1}_{k+1}}+\frac{-1}{w^{1}_{k}-w^{2}_{k+1}}\\
      0&=\frac{2}{w^{2}_{k}-w^{1}_{k}}+\frac{-1}{w^{2}_{k}-w^{1}_{k-1}}+\frac{-1}{w^{2}_{k}-w^{2}_{k-1}}+\frac{-1}{w^{2}_{k}-w^{1}_{k+1}}+\frac{-1}{w^{2}_{k}-w^{2}_{k+1}}\\
      0&=\frac{2}{w^{1}_{n-2}-w^{2}_{n-2}}+\frac{-1}{w^{1}_{n-2}-w^{1}_{n-3}}+\frac{-1}{w^{1}_{n-2}-w^{2}_{n-3}}\\&+\frac{-1}{w^{1}_{n-2}-w_{n-1}}+\frac{-1}{w^{1}_{n-2}-w_{n}}\\
      0&=\frac{2}{w^{2}_{n-2}-w^{1}_{n-2}}+\frac{-1}{w^{2}_{n-2}-w^{1}_{n-3}}+\frac{-1}{w^{2}_{n-2}-w^{2}_{n-3}}\\&+\frac{-1}{w^{2}_{n-2}-w_{n-1}}+\frac{-1}{w^{2}_{n-2}-w_{n}}\\
      0&=\frac{-1}{w_{n-1}-w^{1}_{n-2}}+\frac{-1}{w_{n-1}-w^{2}_{n-2}}\\
      0&=\frac{-1}{w_{n}-w^{1}_{n-2}}+\frac{-1}{w_{n}-w^{2}_{n-2}},\\
    \end{aligned}
  \right.
\end{equation}
where $2\leq k\leq n-3$. By lemma \ref{rlemma},
the solution is as following: $$\bar{w}_{1}=\bar{w}_{k}=\bar{w}_{n-2}=z_{1}+z_{2},w_{n-1}=w_{n}=\frac{z_{1}+z_{2}}{2}$$ and $$w^{1}_{l}w^{2}_{l}=z_{1}z_{2}+\frac{(z_{1}-z_{2})^{2}l(2n-1-l)}{4n(n-1)},\quad 1\leq l\leq n-2.$$ The distribution of the coordinates of critical point is shown in figure \ref{fig:nsbd}. Similar to the case of one singularity in figure  \ref{fig:dd}, the coordinates are symmetric on the same line connecting $z_{1}$ and $z_{2}$, which gives coordinates of the thimble and singularities $z_{1}$ and $z_{2}$ a total order:$$z_{1}<w^{1}_{1}<w^{1}_{k}<w^{1}_{k+1}<w^{1}_{n}=w^{2}_{n}<w^{2}_{k+1}<w^{2}_{k}<w^{2}_{1}<z_{2}.$$
Because of the existence of the thimble above, there is a wall-crossing whenever $o(i)\neq 0$ and its homotopic class will keep this order.

\begin{figure}
 \centering\includegraphics[width=3cm]{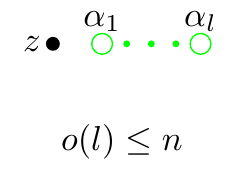}
 \centering\includegraphics[width=3.9cm]{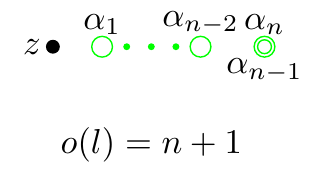}\\
\centering\includegraphics[width=6.5cm]{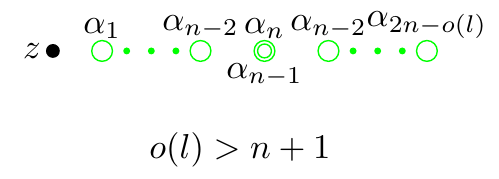}\\
\caption{Coordinates distribution of the $D_{n}$ critical point on ${\bm W}_{c}({\bm w},{\bm z},{\bm \lambda},{\bm l})$ plane near $z$ .}
 \label{fig:dd}
 \end{figure}

Firstly, we consider the case $v_{\lambda^{i}}\otimes v_{\lambda^{j}}\in V_{\omega_{1}}\otimes V_{\omega_{1}}, \quad o(i)+o(j)\neq 2n-1$.

i)$o(i)+o(j)\neq 2n-1\& o(i)\geq o(j) $

There is no wall-crossing, $$\boldsymbol{B}J_{i,j}=q^{-\frac{1}{2}(\lambda^{i},\lambda^{i})}J_{j,i}.$$

ii)$o(i)+o(j)\neq 2n-1\& i<j \& i\neq n$

The variation is of type $II$.
\begin{equation}
\begin{split}
  &\boldsymbol{B}J_{i,j}\\
=&q^{-\frac{1}{2}(\lambda^{i},\lambda^{i})}(J_{j,i}+q^{\frac{1}{2}(\lambda^{i},\lambda^{i}-\lambda^{j})}(q^{-(\lambda^{i},\lambda^{i}-\lambda^{j})}-1)J_{i,j})\\
=&J_{j,i}+(q^{-\frac{1}{2}}-q^{\frac{1}{2}})J_{i,j}.\\
\end{split}
\end{equation}

iii)$o(i)+o(j)\neq 2n-1\& i=n\&j=n+1 $

The variation is of type $I$. As is shown in figure \ref{fig:dbt1}, $$b_{1}+b_{2}=a\cdot q^{\frac{1}{2}(\lambda^{n-1},\alpha_{n})}(1+q^{-\frac{1}{2}(\alpha_{n},\alpha_{n})})$$ is the sum of phase factors of moving one of two $\alpha_{n}$ to the right of $z_{1}$. $$c_{1}+c_{2}=(b_{1}+b_{2})\cdot q^{-(\lambda^{n-1},\alpha_{n})+\frac{1}{2}(\alpha_{n},\alpha_{n})}$$
is from anti-clockwise rotation of $\alpha_{n}$, $2\pi$ around $\{z_{1},\alpha_{1}, ... , \alpha_{n-1}\}$ and $\pi$ around $\alpha_{n}$.
\begin{equation}
\begin{split}
  &\boldsymbol{B}J_{i,j}\\
=&aJ_{j,i}+(c_{1}+c_{2}-b_{1}-b_{2})J_{i,j}\\
=&q^{-\frac{1}{2}(\lambda^{i},\lambda^{i})}(J_{j,i}+q^{\frac{1}{2}(\lambda^{n-1},\alpha_{n})}(1+q^{-\frac{1}{2}(\alpha_{n},\alpha_{n})})(q^{-(\lambda^{n-1},\alpha_{n})+\frac{1}{2}(\alpha_{n},\alpha_{n})}-1)J_{i,j})\\
=&J_{j,i}+(q^{-\frac{1}{2}}-q^{\frac{1}{2}})J_{i,j}.\\
\end{split}
\end{equation}

iv)$o(i)+o(j)\neq 2n-1\& i=n\&j>n+1 $

The variation is a combination of type $I$ and $II$. As shown in figure \ref{fig:dbt12}, there are two possible ways of moving $\alpha_{n}$ from $z_{2}$ to $z_{1}$. But at this time, $\alpha_{n}$ is followed by $\{\alpha_{n-1}, ..., \alpha_{2n+1-j}\}$.
\begin{equation}
\begin{split}
a=&q^{-\frac{1}{2}(\lambda^{i},\lambda^{i})};\\
b_{1}+b_{2}=&a\cdot q^{\frac{1}{2}(\lambda^{n-1},\alpha_{n})+\frac{1}{2}(\lambda^{n},\lambda^{n+1}-\lambda^{j})}(1+q^{-\frac{1}{2}(\alpha_{n},\alpha_{n})});\\ c_{1}+c_{2}=&(b_{1}+b_{2})\cdot q^{-(\lambda^{n-1},\alpha_{n})+\frac{1}{2}(\alpha_{n},\alpha_{n})-(\lambda^{n},\lambda^{n+1}-\lambda_{j})};\\
\boldsymbol{B}J_{i,j}
=&aJ_{j,i}+(c_{1}+c_{2}-b_{1}-b_{2})J_{i,j}\\
=&J_{j,i}+(q^{-\frac{1}{2}}-q^{\frac{1}{2}})J_{i,j}.\\
\end{split}
\end{equation}

In sum,

\begin{equation}
  \boldsymbol{B}J_{i,j}=\left\{
             \begin{array}{ll}
               q^{-\frac{1}{2}}J_{j,i}, & \hbox{$i+j\neq 2n-1\&i=j$;} \\
               J_{j,i}, & \hbox{$i+j\neq 2n-1\&i> j$;} \\
               J_{j,i}+(q^{-\frac{1}{2}}-q^{\frac{1}{2}})J_{i,j}, & \hbox{$i+j\neq 2n-1\&i<j$.}
             \end{array}
           \right.
\end{equation}

For $v_{\lambda^{i}}\otimes v_{\lambda^{j}}\in V_{\omega_{1}}\otimes V_{\omega_{1}}, \quad o(i)+o(j)= 2n-1$, the coefficients of equation \eqref{equ:cof} are $$\boldsymbol{B}_{o^{-1}(2n-1-o(i)),i}^{i,o^{-1}(2n-1-o(i))}=q^{-\frac{1}{2}(\lambda^{o^{-1}(2n-1-o(i))},\lambda^{i})}=q^{\frac{1}{2}}.$$

Let $j>0$ be the difference of the order, then $\boldsymbol{B}_{o^{-1}(2n-1-o(i)),i}^{o^{-1}(o(i)-j),o^{-1}(2n-1-o(i)+j)}$ can be computed as in the following cases.

i) $0<j\leq o(i)<n$ or $o(i)>n\&o(i)-j\geq n$

The variation is of type $III$ and $j$ is the number of moving primary roots. \begin{equation}
\begin{split}
a=&q^{-\frac{1}{2}(\lambda^{o^{-1}(2n-1-o(i))},\lambda^{i})};\\
b=&a\cdot q^{\frac{1}{4}[(\lambda^{o^{-1}(o(i)-j)},\lambda^{o^{-1}(o(i)-j)}-\lambda^{i})+(\lambda^{o^{-1}(2n-1-o(i))},\lambda^{o^{-1}(o(i)-j)}-\lambda^{i})]+\frac{j-1}{2}};\\
c=&b\cdot q^{-(\lambda^{o^{-1}(2n-1-o(i))},\lambda^{o^{-1}(o(i)-j)}-\lambda^{i})}.\\
\end{split}
\end{equation}
$$\boldsymbol{B}_{o^{-1}(2n-1-o(i)),i}^{o^{-1}(o(i)-j),o^{-1}(2n-1-o(i)+j)}
=(-1)^{j}q^{\frac{j+1}{2}}(1-q^{-1}).$$

ii) $n<o(i)\leq 2n-1\& o(i)-j<n-1\& 2o(i)-j\neq 2n-1$

The variation is of type $III$. At this time, the number of primary roots moving from $z_{2}$ to $z_{1}$ is not $j$, but $j-1$. The sign before $b$ is $(-1)^{j-1}$ and the phase factor from $-\pi$ self rotation of $j-1$ primary roots equals to $q^{\frac{j-2}{2}}$. Thus,
\begin{equation}
\begin{split}
b=&a\cdot q^{\frac{1}{4}[(\lambda^{o^{-1}(o(i)-j)},\lambda^{o^{-1}(o(i)-j)}-\lambda^{i})+(\lambda^{o^{-1}(2n-1-o(i))},\lambda^{o^{-1}(o(i)-j)}-\lambda^{i})]+\frac{j-2}{2}};\\
c=&b\cdot q^{-(\lambda^{o^{-1}(2n-1-o(i))},\lambda^{o^{-1}(o(i)-j)}-\lambda^{i})}.\\
\end{split}
\end{equation}
$$\boldsymbol{B}_{o^{-1}(2n-1-o(i)),i}^{o^{-1}(o(i)-j),o^{-1}(2n-1-o(i)+j)}
=(-1)^{j-1}q^{\frac{j}{2}}(1-q^{-1}).$$

iii) $n<o(i)\& o(i)-j<n-1\& 2o(i)-j= 2n-1$

The variation is of type $IV$. The number of primary roots moving from $z_{2}$ to $z_{1}$ is $j-1$. Thus,
\begin{equation}
 \begin{split}
b=&a q^{\frac{1}{4}[(\lambda^{o^{-1}(o(i)-j)},\lambda^{o^{-1}(o(i)-j)}-\lambda^{i})+(\lambda^{o^{-1}(2n-1-o(i))},\lambda^{o^{-1}(o(i)-j)}-\lambda^{i})]+n-o(i)+j-2};\\
c_{1}=&b\cdot q^{-(\lambda^{o^{-1}(2n-1-o(i))},\lambda^{o^{-1}(o(i)-j)}-\lambda^{o^{-1}(o(i)-j+1)})};\\
c_{2}
=&a\cdot q^{\frac{1}{4}[(\lambda^{o^{-1}(o(i)-j)},\lambda^{o^{-1}(o(i)-j)}-\lambda^{i})+(\lambda^{o^{-1}(2n-1-o(i))},\lambda^{o^{-1}(o(i)-j)}-\lambda^{i})]}\\
&q^{{-(\lambda^{o^{-1}(2n-1-o(i))},\lambda^{o^{-1}(o(i)-j)}-\lambda^{i})+(\lambda^{o^{-1}(2n-1-o(i))},\lambda^{o^{-1}(o(i)-j)}-\lambda^{o^{-1}(o(i)-j+1)})}};\\
c_{3}=&a\cdot q^{\frac{1}{4}[(\lambda^{o^{-1}(o(i)-j)},\lambda^{o^{-1}(o(i)-j)}-\lambda^{i})+(\lambda^{o^{-1}(2n-1-o(i))},\lambda^{o^{-1}(o(i)-j)}-\lambda^{i})]}\\
&q^{-(\lambda^{o^{-1}(2n-1-o(i))},\lambda^{o^{-1}(o(i)-j)}-\lambda^{i})}.\\
\end{split}
\end{equation}
$$\boldsymbol{B}_{o^{-1}(2n-1-o(i)),i}^{o^{-1}(o(i)-j),o^{-1}(2n-1-o(i)+j)}
=(-1)^{j-1}q^{\frac{j}{2}}(1-q^{-1})+q^{-\frac{1}{2}}(1-q).$$

iv)$n<o(i)\& o(i)-j=n-1$

The variation is of type $III$. The number of primary roots moving from $z_{2}$ to $z_{1}$ is $j-1$. Thus,
\begin{equation}
\begin{split}
b=&a\cdot q^{\frac{1}{4}[(\lambda^{o^{-1}(o(i)-j)},\lambda^{o^{-1}(o(i)-j)}-\lambda^{i})+(\lambda^{o^{-1}(2n-1-o(i))},\lambda^{o^{-1}(o(i)-j)}-\lambda^{i})]+\frac{j-2}{2}};\\
c=&b\cdot q^{-(\lambda^{o^{-1}(2n-1-o(i))},\lambda^{o^{-1}(o(i)-j)}-\lambda^{i})}.\\
\end{split}
\end{equation}
\begin{equation}
\begin{split}
&\boldsymbol{B}_{o^{-1}(2n-1-o(i)),i}^{o^{-1}(o(i)-j),o^{-1}(2n-1-o(i)+j)}\\
=&\boldsymbol{B}_{o^{-1}(2n-1-o(i)),i}^{n-1,n}\\
=&(-1)^{j-1}q^{\frac{j}{2}}(1-q^{-1}).\\
\end{split}
\end{equation}

v)$o(i)=n\& j\geq 2$

The variation is of type $III$. The number of primary roots moving from $z_{2}$ to $z_{1}$ is $j-1$. Thus,
\begin{equation}
\begin{split}
b=&a\cdot q^{\frac{1}{4}[(\lambda^{o^{-1}(o(i)-j)},\lambda^{o^{-1}(o(i)-j)}-\lambda^{i})+(\lambda^{o^{-1}(2n-1-o(i))},\lambda^{o^{-1}(o(i)-j)}-\lambda^{i})]+\frac{j-2}{2}};\\
c=&b\cdot q^{-(\lambda^{o^{-1}(2n-1-o(i))},\lambda^{o^{-1}(o(i)-j)}-\lambda^{i})}.\\
\end{split}
\end{equation}
$$\boldsymbol{B}_{o^{-1}(2n-1-o(i)),i}^{o^{-1}(o(i)-j),o^{-1}(2n-1-o(i)+j)}
=(-1)^{j-1}q^{\frac{j}{2}}(1-q^{-1}).$$

\begin{equation}\label{Dn1}
\begin{split}
  &\boldsymbol{B}^{o^{-1}(o(i)-j),o^{-1}(2n-1-o(i)+j)}_{o^{-1}(2n-1-o(i)),i}\\
=&\left\{
                                                                 \begin{array}{ll}
                                                                   q^{\frac{j+1}{2}}(-1)^{j}(1-q^{-1}), & \hbox{$o(i)>n\&o(i)-j\geq n\hbox{ or  }0<j\leq o(i)<n$;} \\
                                                                   q^{\frac{j}{2}}(-1)^{j-1}(1-q^{-1}), & \hbox{$o(i)>n\&o(i)-j<n-1\&2o(i)-j\neq 2n-1$;} \\
                                                                   q^{\frac{j}{2}}(-1)^{j-1}(1-q^{-1})\\
+q^{-\frac{1}{2}}-q^{\frac{1}{2}}, & \hbox{$o(i)>n\&o(i)-j<n-1\&2o(i)-j= 2n-1$;} \\
                                                                   q^{\frac{j}{2}}(-1)^{j-1}(1-q^{-1}), & \hbox{$o(i)>n\&o(i)-j=n-1$;} \\
                                                                   q^{\frac{j}{2}}(-1)^{j-1}(1-q^{-1}), & \hbox{$o(i)=n\&j\geq 2$.}
                                                                 \end{array}
                                                               \right.
\end{split}
\end{equation}

\begin{equation}\label{Dn2}
\begin{split}
  &\boldsymbol{B}^{o^{-1}(o(i)-j),o^{-1}(2n-1-o(i)+j)}_{o^{-1}(2n-1-o(i)),i}\\
=&\left\{
                                                                 \begin{array}{ll}
                                                                   q^{\frac{j+1}{2}}(-1)^{j}(1-q^{-1}), & \hbox{$o(i)>n\&o(i)-j\geq n\hbox{ or  }0<j\leq o(i)<n$;}  \\
                                                                   q^{\frac{j}{2}}(-1)^{j-1}(1-q^{-1})\\
+\delta_{2o(i)-j,2n-1}(q^{-\frac{1}{2}}-q^{\frac{1}{2}}), & \hbox{$o(i)\geq n\&o(i)-j\leq n-1$;}
                                                                 \end{array}
                                                               \right.
\end{split}
\end{equation}

Monodromy representation is as following:

\begin{equation}\label{Dn3}
\begin{split}
  &\boldsymbol{B}_{o^{-1}(2n-1-o(a)),a}^{o^{-1}(2n-1-o(b)),b}\\
=&\left\{
                                                  \begin{array}{ll}
                                                    (-1)^{o(a)+o(b)-1}q^{-n+\frac{o(a)+o(b)+1}{2}}(q^{\frac{1}{2}}-q^{-\frac{1}{2}}), & \hbox{$o(a)\leq n-1\& o(b)\geq n$;} \\
&\hbox{$o(a)> n\& o(b)\leq n-1$;} \\
                                                     (-1)^{o(a)+o(b)-2}q^{-n+\frac{o(a)+o(b)}{2}}(q^{\frac{1}{2}}-q^{-\frac{1}{2}})\\
                                                      -\delta_{o(a),o(b)}(q^{\frac{1}{2}}-q^{-\frac{1}{2}}), & \hbox{$o(a)\geq n\& o(b)\geq n$.}
                                                  \end{array}
                                                \right.
\end{split}
\end{equation}

\subsection{Similarity transformation $Q$}
By comparing the monodromy representation with the braid group representation induced by the universal R-matrices (see 7.3B and 7.3C of \cite{Chari}) of quantum group $U_{h}(g)$, we have $\boldsymbol{B_{YY}}=Q\cdot \boldsymbol{B_{U_{h}(g)}}\cdot Q^{-1}$, $Q\in End(V_{\omega_{1}}\otimes V_{\omega_{1}})$.

For $A_{n}$, $Q=Id$;

For $B_{n}$, $Q=\left(
                  \begin{array}{ccc}
                    Id &  &  \\
                     & Q_{2n+1} &  \\
                     &  & Id \\
                  \end{array}
                \right)
$, where \begin{equation}
Q_{2n+1}
=\left(
                     \begin{array}{ccccccccc}
                (-1)^{0}&    &  &   &   &   &    \\
                        &   ...&  &  &  &  &    \\
                        &    & (-1)^{n-1} &   &  &  &    \\
                        &    &  & \frac{(-1)^{n}}{q^{\frac{1}{4}}+q^{-\frac{1}{4}}} &  &  &    \\
                        &    &  &  &  (-1)^{n+1}&  &    \\
                        &    &  &  &    &... &      \\
                        &    &  &  &  &  & (-1)^{2n}   \\
                     \end{array}
                   \right);
\end{equation}

For $C_{n}$ and $D_{n}$, $Q=\left(
                  \begin{array}{ccc}
                    Id &  &  \\
                     & Q_{2n} &  \\
                     &  & Id \\
                  \end{array}
                \right)
$, where\begin{equation}
Q_{2n}
=\left(
                     \begin{array}{ccccccccc}
                (-1)^{0}&    &  &   &   &   &    \\
                        &   ...&  &  &  &  &    \\
                        &    & (-1)^{n-1} &   &  &  &    \\
                        &    &  &   (-1)^{n-1}&  &    \\
                        &    &  &  &   ... &      \\
                        &    &  &  &  &  (-1)^{2n-2}   \\
                     \end{array}
                   \right).
\end{equation}
This proves the first main theorem.

\section{Commutation of two parameter deformations}\label{sec:conclusion}

In the previous section, we have studied the transformation $T(1)$ of parameters $z_{1}$ and $z_{2}$. In this section, we consider the continuous deformation of the real parameter $c$ and its relation with $T(1)$.
As shown in theorem \ref{1sing}, for the fundamental representation $V_{\omega_{1}}$ of each classical complex simple Lie algebra, the lowest weight vectors in $SingV_{\omega_{1}}\otimes V_{\omega_{1}}$ are
\begin{equation}
\begin{split}
A_{n}:&\quad v_{2\omega_{1}-\alpha_{1}}\\
B_{n}:&\quad v_{2\omega_{1}-2\alpha_{1}-...-2\alpha_{n}}\\
C_{n}:&\quad v_{2\omega_{1}-2\alpha_{1}-...-2\alpha_{n-1}-\alpha_{n}}\\
D_{n}:&\quad v_{2\omega_{1}-2\alpha_{1}-...-2\alpha_{n-2}-\alpha_{n-1}-\alpha_{n}}.
\end{split}
\end{equation}
For $A_{n}$, the Yang-Yang function with $c=0$ corresponding to the singular vector $v_{2\omega_{1}-\alpha_{1}}$ is
$${\bm W}({\bm w},{\bm z},{\bm \lambda},1)=\sum_{a=1}^{2}(\alpha_{1},\omega_{1})\log(w_{1}-z_{a})-(\omega_{1},\omega_{1})\log(z_{1}-z_{2}).$$
Its unique critical point is $w_{1}=\frac{z_{1}+z_{2}}{2}$.
For $B_{n}$, $C_{n}$ and $D_{n}$, the corresponding critical point equations \eqref{equ:BNSB} \eqref{equ:CNSB} \eqref{equ:DNSB} are already solved.
Thus, corresponding to the lowest weight vector in $SingV_{\omega_{1}}\otimes V_{\omega_{1}}$ of each classical Lie algebra, there exists an unique thimble $J\subset \mathfrak{J}_{0}(P)$ for $\|e^{-\frac{{\bm W}}{k+h^{v}}}\|$ connecting $z_{1}$ and $z_{2}$ .
The distributions of the coordinates of the critical point on ${\bm W}$ plane are shown in figures \ref{fig:nsba}, \ref{fig:nsbb}, \ref{fig:nsbc} and \ref{fig:nsbd} respectively. As shown in section \ref{subsct:bundle}, $\mathfrak{J}_{0}(P)$ is a $\mathbb{Z}[t,t^{-1}]$-module. Thus, $J$ naturally generates a one dimensional $\mathbb{Z}[t,t^{-1}]$-submodule of $\mathfrak{J}_{0}(P)$. We denote also by $\boldsymbol{B_{YY}}$ the monodromy representation induced by $T(1)$ on the one dimensional $\mathbb{Z}[t,t^{-1}]$-submodule generated by $J$.

\begin{figure}
 \centering\includegraphics[width=2.5cm]{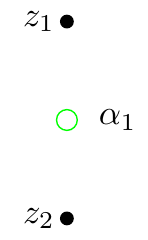}\\
\caption{Coordinates distribution of the $A_{n}$ critical point on ${\bm W}$ plane at $c=0$.}
 \label{fig:nsba}
 \end{figure}

\begin{figure}
 \centering\includegraphics[width=3.5cm]{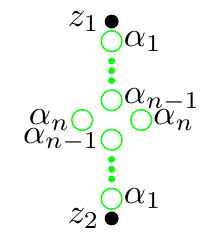}\\
\caption{Coordinates distribution of the $B_{n}$critical point on ${\bm W}$ plane at $c=0$.}
 \label{fig:nsbb}
 \end{figure}

\begin{figure}
 \centering\includegraphics[width=3cm]{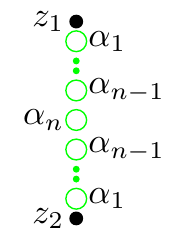}\\
\caption{Coordinates distribution of the $C_{n}$ critical point on ${\bm W}$ plane at $c=0$.}
 \label{fig:nsbc}
 \end{figure}

\begin{figure}
\centering\includegraphics[width=4cm]{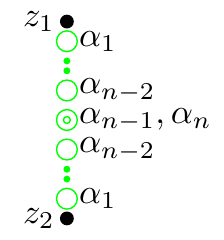}\\
\caption{Coordinates distribution of the $D_{n}$ critical point on ${\bm W}$ plane at $c=0$.}
 \label{fig:nsbd}
 \end{figure}

Let $\alpha_{I}=\{\alpha_{i_{1}},\alpha_{i_{2}}, ...,\alpha_{i_{m}}\}$ be the set of all the primary roots in the lowest weight $2\omega_{1}-\alpha_{i_{1}}-\alpha_{i_{2}}...-\alpha_{i_{m}}$ and $I=\{i_{1},i_{2},...,i_{m}\}$ the set of their indexes. Define $SI(\alpha_{I})=\sum_{i_{j},i_{s}\in I\&j<s}(\alpha_{i_{j}},\alpha_{i_{s}})$ as the sum of the pairwise inner products of the primary roots in  $\alpha_{I}$. By a direct calculation for the phase factor difference from the rotation $T(1)$, we have

\begin{lemma}
\begin{equation}\label{equ:c0}
  \boldsymbol{B_{YY}}J={\bm d} J, \quad {\bm d}=q^{-\frac{1}{2}[(\omega_{1},\omega_{1})-2(\omega_{1},\sum_{i_{j}\in I} \alpha_{i_{j}})+SI(\alpha_{I})]}.
\end{equation}
For $A_{n}$ Lie algebra, $m=1$, $i_{1}=1$, $$SI(\alpha_{I})=0.$$
For $B_{n}$ Lie algebra, $m=2n$, $i_{j}=\left\{
                                          \begin{array}{ll}
                                            j, & \hbox{$j\leq n$;} \\
                                            2n+1-j, & \hbox{$j>n$,}
                                          \end{array}
                                        \right.
$ $$SI(\alpha_{I})=-2(n-1)+1=-2n+3.$$
For $C_{n}$ Lie algebra, $m=2n-1$, $i_{j}=\left\{
                                          \begin{array}{ll}
                                            j, & \hbox{$j\leq n$;} \\
                                            2n-j, & \hbox{$j>n$,}
                                          \end{array}
                                        \right.$ $$SI(\alpha_{I})=-(n-2)-1=-n+1.$$
For $D_{n}$ Lie algebra, $m=2n-2$, $i_{j}=\left\{
                                          \begin{array}{ll}
                                            j, & \hbox{$j\leq n$;} \\
                                            2n-1-j, & \hbox{$j>n$,}
                                          \end{array}
                                        \right.$
$$SI(\alpha_{I})=-2(n-1)=-2n+2.$$
\end{lemma}
Although theorem \ref{1prm} demands that $c\in \mathbb{Z}_{\geq0}$, the thimble structure can be defined continuously on $c\geq0$, thus we can consider the continuous deformation of $c\rightarrow+\infty$ from $c=0$.
For any $J\subset \mathfrak{J}_{0}(P)$, when $c\rightarrow +\infty$, by lemma \ref{split}, the coordinates of its critical point tend either to $z_{1}$ or to $z_{2}$, thus it gives several possible thimbles in $J_{a,b}\subset \mathfrak{J}(P)$.
Multiply to each thimble $J_{a,b}$ the phase factor generated from the deformation and sum them up. Define the symmetry breaking transformation $\boldsymbol{S}$ by
$$\boldsymbol{S}:\mathfrak{J}_{0}(P)\rightarrow\mathfrak{J}(P)$$
\begin{equation}\label{equ:bt}
  \boldsymbol{S} J=\sum_{a,b}e^{a,b}J_{a,b},
\end{equation}

where coefficients $e^{a,b}$ are the phase factor difference of $e^{-\frac{{\bm W}({\bm w},{\bm z},{\bm \lambda},{\bm l})}{\kappa+h^{\vee}}}$ in the process of $c\rightarrow +\infty$ from $0$. For $J\subset \mathfrak{J}_{0}(P)$ corresponding to the lowest weight vector in $Sing V_{\omega_{1}}\otimes V_{\omega_{1}}$, by a direct calculation, we have

\begin{lemma}
For $A_{n}$,
$$\boldsymbol{S} J=q^{-\frac{1}{4}(\omega_{1},\alpha_{1})}J_{n,0}+q^{\frac{1}{4}(\omega_{1},\alpha_{1})}J_{0,n}=q^{-\frac{1}{4}}J_{n,0}+q^{\frac{1}{4}}J_{0,n}.$$
For $B_{n}$,
\begin{equation}
  \begin{split}
  \boldsymbol{S} J=&\sum_{i<n}(-1)^{i}q^{-\frac{n-i}{2}+\frac{1}{4}}J_{2n-i,i}+(-1)^{n}(q^{-\frac{1}{4}+\frac{1}{4}})J_{n,n}\\
  &+\sum_{i>n}(-1)^{i}q^{-\frac{n-i}{2}-\frac{1}{4}}J_{2n-i,i}.
  \end{split}
\end{equation}
For $C_{n}$, $$\boldsymbol{S} J=\sum_{i<n}(-1)^{i}q^{-\frac{n-i}{4}}J_{2n-1-i,i}+\sum_{i\geq n}(-1)^{i}q^{-\frac{n-i-1}{4}}J_{2n-1-i,i}.$$
For $D_{n}$,
\begin{equation}
  \begin{split}
\boldsymbol{S} J=&\sum_{i<n-1}(-1)^{i}q^{-\frac{n-i-1}{2}}J_{2n-2-i,i}+\sum_{i\geq n}^{2n-2}(-1)^{i}q^{-\frac{n-1-i}{2}}J_{2n-2-i,i}\\
&+(-1)^{n-1}(J_{n-1,n-1'}+J_{n-1',n-1}).
\end{split}
\end{equation}
\end{lemma}

By the monodromy representation $\boldsymbol{B}_{a,b}^{c,d}$ in section \ref{subsct:monodromy} and the lemma above, it is straight forward to prove that $e^{a,b}$ is the eigenvector of the monodromy representation.
\begin{lemma}
\begin{equation}\label{eigenvector}
  \sum_{a,b}e^{a,b}\boldsymbol{B}^{c,d}_{a,b}={\bm d}e^{c,d},
\end{equation}
where ${\bm d}$ is defined in \eqref{equ:c0}.
\end{lemma}
By this lemma, $$\boldsymbol{B_{YY}} \boldsymbol{S} J=\boldsymbol{B_{YY}}\sum_{a,b}e^{a,b}J_{a,b}={\bm d}\sum_{a,b,c,d}e^{a,b}\boldsymbol{B}^{c,d}_{a,b}J_{c,d}=\boldsymbol{S}\boldsymbol{B_{YY}}J.$$ The second main theorem \ref{thm2} follows.

From the equation \eqref{equ:bt}, it is natural to define a $ \mathbb{Z}[t,t^{-1}]$ linear operator $$M: \mathbb{Z}[t,t^{-1}]\rightarrow V_{\omega_{1}}\otimes V_{\omega_{1}},$$ satisfying $$M(1)=\sum_{a,b}e^{a,b}v_{\lambda^{a}}\otimes v_{\lambda^{b}}.$$

We can also define a creation matrix associated with $M$ as
$$M^{a,b}=e^{a,b}.$$
Annihilation matrix $M_{a,b}$ is its inverse satisfying
$$\sum_{b}M_{a,b}M^{b,c}=\delta_{a}^{c}.$$ With these two matrices, quantum trace of any $(m,m)$ tensor $T_{i_{1}i_{2}...i_{m}}^{j_{1}j_{2}...j_{m}}$ is just as following:
$$Tr_{q}T=\sum_{i_{1},i_{2},...i_{m},j_{1},j_{2},...,j_{m}}T_{i_{1}i_{2}...i_{m}}^{j_{1}j_{2}...j_{m}}\eta_{j_{1}}^{i_{1}}\eta_{j_{2}}^{i_{2}}...\eta_{j_{m}}^{i_{m}},$$
where $\eta^{i_{k}}_{j_{k}}=\sum_{l}M^{i_{k},l}M_{l,j_{k}}$. By Alexander's theorem (p91 I.7 of \cite{Kauffman}), the ambient knots invariants defined by contraction of $\boldsymbol{B}_{a,b}^{c,d}$, $M_{a,b}$ and $M^{a,b}$ in the decomposition of some knot projection diagram coincide with the quantum trace of the $(m,m)$ tensor associated with the braid. In \cite{HL1} and \cite{HL2}, knots invariants associated with the fundamental representation of $A_{n}$ Lie algebra and $B_{n},C_{n},D_{n}$ Lie algebra are proved to be HOMFLY-PT polynomial and Kauffman polynomial respectively.

\begin{remark} Knots invariant depends on representations of Lie algebra. Different representations may give different knots invariants. Generally, the corresponding knots invariant is not necessary to be  HOMFLY-PT polynomial or Kauffman polynomial. We will focus on this point elsewhere.
\end{remark}

\bibliography{Yang-Yang}
\bibliographystyle{unsrt}
\end{document}